\documentclass[11pt,journal,onecolumn]{IEEEtran}

\usepackage{authblk,tikz}
\usepackage[style=ieee,backend=bibtex8,doi=false,isbn=false,url=false]{biblatex}
\usepackage{amsmath,mathtools,amssymb,amsfonts,dsfont,amsthm,mathrsfs,multicol,array,etoolbox} 
\usepackage{graphicx,wrapfig} 
\usepackage{hyperref} 
\usepackage{enumerate}

\DeclarePairedDelimiterX{\ket}[1]{\lvert}{\rangle}{#1}
\DeclarePairedDelimiterX{\kket}[1]{\lvert}{\rangle\rangle}{#1}
\DeclarePairedDelimiterX{\bra}[1]{\langle}{\rvert}{#1}
\DeclarePairedDelimiterX{\bbra}[1]{\langle\langle}{\rvert}{#1}
\newcommand{\bracket}[2]{\langle #1 \vert #2 \rangle}

\newcommand{\abs}[1]{\left\lvert #1 \right\rvert} 
\newcommand{\tr}{\operatorname{Tr}} 

\newcommand{\pnorm}[2][p]{\left\lVert #2 \right\rVert_{#1}} 
\newcommand{\dnorm}[1]{\left\lVert #1 \right\rVert_{\diamond}} 

\newcommand{\cptp}{\text{CPTP}} 

\newcommand{\sepd}{\text{SepD}} 
\newcommand{\sepc}{\text{SepC}} 
\newcommand{\seppsc}{\text{SEPPSC}} 

\makeatletter
\def\thmhead@plain#1#2#3{%
  \thmname{#1}\thmnumber{\@ifnotempty{#1}{ }\@upn{#2}}%
  \thmnote{ {\the\thm@notefont#3}}}
\let\thmhead\thmhead@plain
\makeatother

\theoremstyle{definition}
\newtheorem{definition}{Definition}[section]
\theoremstyle{definition}
\newtheorem{theorem}{Theorem}[section]

\newtheorem{lemma}[theorem]{Lemma}
\theoremstyle{remark}
\newtheorem{proposition}[theorem]{Proposition}
\theoremstyle{remark}

\mathchardef\mhyphen="2D

\begin{document}

\title{One-Shot Manipulation of Entanglement for Quantum Channels}

\author{Ho-Joon Kim, Soojoon Lee, Ludovico Lami, and Martin B. Plenio
\thanks{H.-J.~Kim and S.~Lee are with the Department of Mathematics, Kyung Hee University, Seoul 02447, Korea. L.~Lami and M.B.~Plenio are with the Institute of Theoretical Physics and IQST, Universit\"at Ulm, Albert-Einstein-Allee 11, D-89081 Ulm, Germany.}}


\maketitle

\begin{abstract}
We show that the dynamic resource theory of quantum entanglement can be formulated using the superchannel theory. In this formulation, we identify the separable channels and the class of free superchannels that preserve channel separability as free resources, and choose the swap channels as dynamic entanglement golden units. Our first result is that the one-shot dynamic entanglement cost of a bipartite quantum channel under the free superchannels is bounded by the standard log-robustness of channels. The one-shot distillable dynamic entanglement of a bipartite quantum channel under the free superchannels is found to be bounded by a resource monotone that we construct from the hypothesis-testing relative entropy of channels with minimization over separable channels. We also address the one-shot catalytic dynamic entanglement cost of a bipartite quantum channel under a larger class of free superchannels that could generate the dynamic entanglement which is asymptotically negligible; it is bounded by the generalized log-robustness of channels.
\end{abstract}

\section{Introduction}
The emergence of the modern development of quantum information science is tightly linked to a fundamental change of paradigm that characterizes our appreciation of fundamental traits of quantum mechanics. Rather than viewing these merely as counter-intuitive departures from our classical world view, in recent years we have come to recognize fundamental quantum features as resources that enable us to solve technological and information theoretic tasks more efficiently than classical physics would allow. The desire to investigate systematically which aspects of quantum mechanics are responsible for potential operational advantages has led to the development of quantum resource theories \cite{chitambar2019quantum}. The most basic concept that gives rise to the structure of resource theories is the concept of constraints that are imposed on our ability to operate beyond those that are already enforced by the laws of quantum mechanics. From this emerges by an elegant inevitability the concept of free states and operations. These are those that can be prepared and executed without violation of the constraints. These two main ingredients allow for the formulation of a rigorous theoretical framework in which to analyze resources quantitatively. Perhaps the most fundamental examples are represented by the theory of quantum coherence, which marks the delineation between classical and quantum physics already at the level of individual particles \cite{baumgratz2014QuantifyingCoherence,winter2016OperationalResourceCoherence,streltsov2017QuantumCoherenceResource}, and, historically having emerged first, the theory of entanglement, which explores the value of quantum correlations as opposed to classical correlations \cite{plenio2007introduction,horodecki2009quantum}. These were followed by a host of resource theories including that of superposition \cite{aberg2006QuantifyingSuperposition,theurer2017ResourceTheorySuperposition}, of reference frames \cite{gour2008resource}, of Gaussianity \cite{lami2018gaussian}, of quantum optical non-classicality \cite{kctan2017QuantifyingCoherenceCoherent,yadin2018operational,ferrari2020asymptotic}, of indistinguishable particles \cite{killoran2014extracting,killoran2016converting,morris2020entanglement},
and of thermodynamics \cite{brandao2013ResourceThermalEquilibrium,horodecki2013FundamentalLimitationsThermodynamics,brandao2015SecondLawsThermodynamics}. 

Initially, the focus of attention in entanglement theory was placed squarely on the entanglement content of quantum states, i.e.\ (i)~which states contain entanglement \cite{Peres1996,horodecki1996}, (ii)~how the entanglement of states can be transformed under local operations and classical communication \cite{bennett1996a, nielsen1999ConditionsEntanglementTransformations, jonathan1999minimal, vidal2000approximate}, (iii)~how entanglement can be verified quantitatively \cite{audenaert2006correlations,eisert2007quantitative,guhne2007estimating,lanyon2017efficient} and (iv)~how useful entanglement is in operational tasks, e.g.\ to enable the realisation of arbitrary non-local quantum operations when only local operations and classical communication is available. For example, maximally entangled states may be employed to achieve general non-local quantum gates between spatially separated parties using only local quantum operations and classical communication \cite{eisert2000OptimalLocalImplementation, collins2001nonlocal}. This example is characteristic of the early approaches to entanglement theory in particular and resource theories in general. While task (ii) concerns state-to-state transformations, task (iv) is of a somewhat different nature, as it connects static resources (states) with dynamic ones (quantum operations).

In fact, quantum states may be considered as special cases of quantum channels, and these subsume also quantum measurements and quantum dynamics~\cite{coecke2016MathematicalTheoryResources, gour2018EntropyQuantumChannel, theurer2019QuantifyingOperationsApplication}. To make this concrete, consider that quantum states and measurements can be thought of as quantum channels with trivial input and classical output, respectively. 

While such approach is legitimate based on the fact that any quantum channel can be simulated with free operations and entangled states, the quantum community is now aiming to encompass all the aspects in a unified manner by studying the properties of quantum channels with modern tools of quantum resource theories. First steps in this direction have been taken with the extension of the entropy of a quantum state to that of  quantum channels with its operational meaning given by the channel merging \cite{gour2018EntropyQuantumChannel}; the entropy of a preparation channel reduces to the usual entanglement of the state it prepares. The resource theory of the coherence of quantum dynamics have been investigated \cite{theurer2019QuantifyingOperationsApplication,zwliu2019ResourceTheoriesQuantum,ycliu2019OperationalResourceTheory,saxena2020DynamicalResourceTheory} 
The properties of a quantum channel associated with entanglement also has been investigated using the tools of resource theories \cite{gour2019EntanglementBipartiteChannel,gour2020DynamicalEntanglement,bauml2019ResourceTheoryEntanglement,xwang2020CostQuantumEntanglement,xwang2018ExactEntanglementCost,theurer2020}. Moreover, there are recent results that built various ways to construct resource monotones of dynamical resources in general which emphasize similarities and subtle differences from the quantum resources in quantum states \cite{gour2019HowQuantifyDynamical}. Alongside, operational interpretations of channel resource theories have been identified \cite{theurer2019QuantifyingOperationsApplication,zwliu2019ResourceTheoriesQuantum,yliu2020OperationalResourceTheory}.

In this paper, we will use as basic building blocks specific maximally resourceful operations that play the role that maximally entangled state played in the entanglement theory of states and explore how concatenation and combination with free operations allows us to achieve general quantum channels. As the basic element is itself an operation rather than a static state, this approach is now running under dynamic entanglement theory. As a note of caution though, this should not be equated with an even more general approach in which one indeed considers continuous in time evolution based on some generators. More specifically, we consider the problem of quantum channel manipulation in the one-shot setting, 
%
when the allowed set of free superchannels is maximal, in the sense that it comprises all transformations that map separable channels to separable channels. We choose the $K$-swap channels as dynamic entanglement resources, which play the role of $K$-maximally entangled states $ \Phi_{A_{0}B_{0}}^{K} $ in the entanglement theory of quantum states. The dynamic entanglement resource is intimately related to the static entanglement resource of quantum states in the sense that the former can generate the maximum static resource under LOCC as well as it requires the maximum static resource to simulate them under LOCC, so it bridges the dynamic entanglement to the maximum static entanglement. In order to treat the dynamic entanglement resource required for conversions between quantum channels, we define the separability-preserving superchannels and use them as the free superchannels. The dynamic entanglement resource of a bipartite quantum channel is investigated in operational ways. To evaluate the one-shot dynamic entanglement cost of a bipartite quantum channel, we ask which amount of dynamic entanglement resource is necessary to simulate the channel by means of free superchannels. To evaluate the one-shot distillable dynamic entanglement of the channel, instead, we determine how much dynamic entanglement resource can be obtained by using only free superchannels. 
We further push the notion of dynamic entanglement cost to its limits, by considering a two-fold variation on the theme: on the one hand, we allow for catalysts, while on the other we define and use a larger set of free superchannels that might generate a small amount of dynamic entanglement. We refer to the resulting modified notion as the one-shot catalytic dynamic entanglement cost of the channel.
Finally, in order to get some insight into the asymptotic scenario, we adopt the liberal smoothing to define liberal dynamic entanglement cost and show that it is given by the liberal regularized relative entropy of channels with respect to the free channels.

\section{Dynamic Entanglement Resource}
We use indexed capital letters such as $A_{0}, B_{1}$, etc.\ to denote physical systems, and juxtapose them to indicate physical composites. $\mathcal{B}(\mathcal{H}_{A_{0}})$ denotes the space of bounded operators acting on a finite dimensional Hilbert space $\mathcal{H}_{A_{0}}$. The set of linear maps from $ \mathcal{B}(\mathcal{H}_{A_{0}})$ to $\mathcal{B}(\mathcal{H}_{A_{1}})$ will be denoted with $\mathcal{L}(A)\equiv \mathcal{L}(A_{0}\to A_{1})$; quantum channels, that is, completely positive trace-preserving linear maps in $\mathcal{L}(A)$, will be collectively denoted with $\cptp(A)\equiv \cptp(A_{0}\to A_{1})$. We use calligraphic letters for quantum channels and abbreviations such as $ \mathcal{E}_{A}\equiv \mathcal{E}_{A_{0}\to A_{1}}$. As an exception, we omit indices if we take the trace map over all the input spaces such as $\tr\left(X_{A_{0}B_{0}}\right)$. We sometimes omit the identity channel for readability when there's little chance of confusion. As a distance between two quantum channels, we use the metric induced by the diamond norm denoted as $\dnorm{\cdot}$ \cite{kitaev1997QuantumComputationsAlgorithms,paulsen2003}. $ \mathbb{L}(A\to A')$ denotes the set of linear supermaps from $ \mathcal{L}(A) $ to $ \mathcal{L}(A') $. A Greek letter $ \Theta_{A\to B} $ is used to denote supermaps, whose action is expressed as $\Theta_{A\to B}[\mathcal{E}_{A}]$. We write $ \Psi_{A_{0}} $ for the density matrix of the pure state $\ket{\Psi}_{A_{0}}$, and call $\mathcal{S}(A_{0})$ and $\mathcal{D}(A_{0})$ the sets of pure and mixed states of system $A_{0}$, respectively. The set of separable states on $ A_{0}B_{0} $ is indicated with $\sepd(A_{0}:B_{0})$, while $ \sepc(A:B)$ stands for the set of separable channels from $ A_{0}B_{0} $ to $ A_{1}B_{1} $.

A $K$-swap channel $\mathcal{F}_{AB}^K$ consists in the application of the $K$-swap gate $F_{AB}^{K} = \sum_{i,j=0}^{K-1}\ket{ij}\!\bra{ji}_{A_{0}B_{0}\to A_{1}B_{1}}$; it is a typical example of a separability-preserving channel\footnote{Also called non-entangling map \cite{BrandaoPlenio2}.} that is not actually separable as a map \cite{harrow2003RobustnessQuantumGates}. We use the $K$-swap channel $ \mathcal{F}_{AB}^{K} $ as the ``golden unit'' of resource in the theory of dynamical entanglement \cite{zwliu2019OneShotOperationalQuantum}; its role is entirely analogous to that of the $K$-maximally entangled state $\ket{\Phi^{K}}_{A_{0}B_{0}} = \frac{1}{\sqrt{K}}\sum_{i=0}^{K-1}\ket{ii}_{A_{0}B_{0}}$ in the static entanglement theory.
Other reasonable choices for the golden unit are entirely equivalent to ours and therefore lead to the same results. Indeed, consider that a $K$-swap channel can generate a pair of $K$-maximally entangled states between Alice and Bob under LOCC (Fig.~\ref{fig:swap channel generating entangled pair}), while they also need two $K$-maximally entangled states to simulate a $K$-swap gate with LOCC.
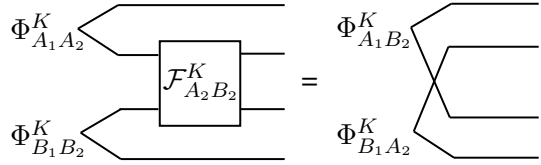
\begin{figure}
	\centering
	\tikzset{every picture/.style={line width=0.75pt}} 
	
	\begin{tikzpicture}[x=0.75pt,y=0.75pt,yscale=-0.7,xscale=0.7]
		
		\draw   (160,126) -- (218,126) -- (218,187) -- (160,187) -- cycle ;
		\draw    (130,136) -- (159,136) ;
		\draw    (131,175) -- (159,175) ;
		\draw    (103,192) -- (132,211) ;
		\draw    (103,192) -- (132,175) ;
		\draw    (218,135) -- (250,135) ;
		\draw    (130,100) -- (250,100) ;
		\draw    (132,211) -- (250,211) ;
		\draw    (101,117) -- (130,136) ;
		\draw    (101,117) -- (130,100) ;
		\draw    (218,175) -- (250,175) ;
		\draw    (368,130) -- (433,130) ;
		\draw    (343,191) -- (372,210) ;
		\draw    (343,191) -- (368,130) ;
		\draw    (370,99) -- (433,99) ;
		\draw    (372,210) -- (433,210) ;
		\draw    (341,116) -- (369,181) ;
		\draw    (341,116) -- (370,99) ;
		\draw    (369,181) -- (433,181) ;
		
		\draw (160,140) node [anchor=north west][inner sep=0.75pt]   [align=left] {$ \mathcal{F}_{A_{2}B_{2}}^{K} $};
		\draw (50,104) node [anchor=north west][inner sep=0.75pt]   [align=left] {$\Phi_{A_{1}A_{2}}^{K}$};
		\draw (50,181) node [anchor=north west][inner sep=0.75pt]   [align=left] {$\Phi_{B_{1}B_{2}}^{K}$};
		\draw (260,150) node [anchor=north west][inner sep=0.75pt]   [align=left] {=};
		\draw (285,103) node [anchor=north west][inner sep=0.75pt]   [align=left] {$\Phi_{A_{1}B_{2}}^{K}$};
		\draw (285,180) node [anchor=north west][inner sep=0.75pt]   [align=left] {$\Phi_{B_{1}A_{2}}^{K}$};
		
	\end{tikzpicture}
	\caption{Two $K$-maximally entangled states generated by the $K$-swap channel from locally prepared $K$-maximally entangled states.}
	\label{fig:swap channel generating entangled pair}
\end{figure}

A superchannel is a linear supermap that sends a quantum channel to another quantum channel \cite{chiribella2008TransformingQuantumOperations,gour2019ComparisonQuantumChannels}. A superchannel can be decomposed into pre- and post-channel with an additional memory system as in Fig.~\ref{fig: structure of a superchannel}.
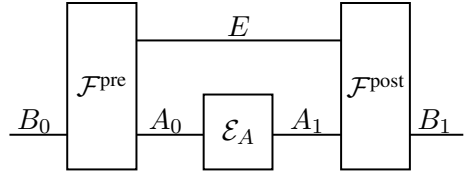
\begin{figure}
	\centering
	\tikzset{every picture/.style={line width=0.75pt}} 
	\begin{tikzpicture}[x=0.75pt,y=0.75pt,yscale=-0.7,xscale=0.7]
		
		\draw   (255.12,100) -- (305,100) -- (305,220) -- (255.12,220) -- cycle ; 
		\draw   (353.42,166) -- (403,166) -- (403,220) -- (353.42,220) -- cycle ; 
		\draw   (452.45,100) -- (502,100) -- (502,220) -- (452.45,220) -- cycle ; 
		
		\draw    (213.5,195) -- (255.12,195) ; 
		
		\draw    (502,195) -- (540.5,195) ; 
		
		\draw    (305,195) -- (353.42,195) ; 
		
		\draw    (403,195) -- (452.45,195) ; 
		
		\draw    (305,127) -- (452.45,127) ; 
		
		\draw (280.23,160) node   {$\mathcal{F}^{\text{pre}}$};
		\draw (477.56,160) node   {$\mathcal{F}^{\text{post}}$};
		\draw (378.54,193) node   {$\mathcal{E}_{A}$};
		\draw (231.44,185) node   {$B_{0}$};
		\draw (326.16,185) node   {$A_{0}$};
		\draw (428.05,185) node   {$A_{1}$};
		\draw (520.0,185) node   {$B_{1}$};
		\draw (378.54,117) node   {$E$};
	\end{tikzpicture}
	\caption{Structure of a superchannel $ \Theta_{A\to B} $ acting on a quantum channel $ \mathcal{E}_{A} $.}
	\label{fig: structure of a superchannel}
\end{figure}
Taking the separable channels as the free channels, we define the largest set of superchannels that does not generate any dynamic entanglement resource:
\begin{definition}
	A superchannel $ \Theta_{AB\to A'B'} $ is called a \emph{separability-preserving superchannel} (SEPPSC) if \begin{equation}
		\Theta_{AB\to A'B'}[ \mathcal{M}_{AB}] \in \sepc(A':B') \quad \forall \mathcal{M}_{AB} \in \sepc(A:B).
	\end{equation}
	The set of all separability-preserving superchannels from $ \mathcal{L}(AB) $ to $ \mathcal{L}(A'B') $ is denoted as SEPPSC($ A:B\to A':B' $).
\end{definition}

As dynamic entanglement monotones, we use the standard and the generalized robustness of channels. Since having been introduced for states~\cite{VidalTarrach, Steiner2003, regula2017ConvexGeometryQuantum}, these measures have found widespread applications to the quantitative analysis of operational tasks, most notably subchannel discrimination~\cite{Takagi2019, lami2020taming, regula2020operational}. The standard robustness of a bipartite quantum channel with respect to the separable channels is defined as
\begin{equation}\label{standard_rob}
	R_{s} (\mathcal{N}_{AB}) \coloneqq \min \left\{ s\ge 0 : \dfrac{\mathcal{N}_{AB} + s \mathcal{M}_{AB}}{1+s}\in \sepc(A:B), \mathcal{M}_{AB}\in \sepc(A:B) \right\},
\end{equation}
while the generalized robustness of a bipartite quantum channel with respect to the separable channels is defined as
\begin{equation}\label{generalized_rob}
	R (\mathcal{N}_{AB}) \coloneqq \min \left\{ s\ge 0 : \dfrac{\mathcal{N}_{AB} + s \mathcal{M}_{AB}}{1+s}\in \sepc(A:B), \mathcal{M}_{AB}\in \cptp(AB) \right\}.
\end{equation}
The standard log-robustness and the generalized log-robustness of channels with respect to the separable channels are given by
\begin{equation}
    LR_{s} (\mathcal{N}_{AB}) \coloneqq \log \left \{ 1+R_{s}(\mathcal{N}_{AB})\right\},\quad LR (\mathcal{N}_{AB}) \coloneqq \log \left \{ 1+R(\mathcal{N}_{AB})\right\},
\end{equation}
respectively, where the logarithm uses base 2. For $\varepsilon \ge 0$, the smooth versions of the above quantities are defined as
\begin{equation}\label{key}
	LR_{s}^{\varepsilon} (\mathcal{N}_{AB}) = \min_{\mathcal{N}_{AB}'\approx_{\varepsilon} \mathcal{N}_{AB}} LR_{s}(\mathcal{N}_{AB}'), \quad LR^{\varepsilon} (\mathcal{N}_{AB}) = \min_{\mathcal{N}_{AB}'\approx_{\varepsilon} \mathcal{N}_{AB}} LR(\mathcal{N}_{AB}'),
\end{equation}
where $ \mathcal{N}_{AB}'\approx_{\varepsilon} \mathcal{N}_{AB} $ is a shorthand for $ \frac{1}{2}\dnorm{\mathcal{N}_{AB}' - \mathcal{N}_{AB}} \le \varepsilon$.
The generalized log-robustness of channels has been found to have an operational meaning in the context of resource erasure \cite{zwliu2019ResourceTheoriesQuantum}; It also can be expressed with the max-relative entropy of channels minimized over the set of the separable channels:
\begin{equation}
    LR^{\varepsilon} (\mathcal{N}_{AB}) = \min_{\mathcal{N}_{AB}'\approx_{\varepsilon} \mathcal{N}_{AB}} \min_{\mathcal{M}_{AB}\in \sepc(A:B)} D_{\max} \left( \mathcal{N}_{AB}'\Vert \mathcal{M}_{AB} \right).
\end{equation}
Both robustnesses are monotonically nonincreasing under the free superchannels:
\begin{lemma}
	For $ \Theta_{AB\to A'B'}\in\seppsc(AB\to A'B') $, it holds that
	\begin{gather*}
		R_{s}(\Theta_{AB\to A'B'}[\mathcal{N}_{AB}])\le R_{s}(\mathcal{N}_{AB}),\\
		R(\Theta_{AB\to A'B'}[\mathcal{N}_{AB}])\le R(\mathcal{N}_{AB}).
	\end{gather*}
\end{lemma}
\begin{proof}
	For $ r = R_{s}(\mathcal{N}_{AB}) $, there exist separable channels $ \mathcal{M}_{AB} $ and $ \mathcal{L}_{AB} $ satisfying that
	\begin{equation}\label{key}
		\mathcal{N}_{AB}+  r \mathcal{M}_{AB} = (1+r) \mathcal{L}_{AB}.
	\end{equation}
	For $ \Theta_{AB\to A'B'}\in\seppsc(AB\to A'B') $, we have that
	\begin{equation}
	\Theta_{AB\to A'B'}[\mathcal{N}_{AB}]+  r \Theta_{AB\to A'B'}[\mathcal{M}_{AB}] = (1+r) \Theta_{AB\to A'B'}[\mathcal{L}_{AB}],
	\end{equation}
	where $\Theta_{AB\to A'B'}[\mathcal{M}_{AB}]$ and $\Theta_{AB\to A'B'}[\mathcal{L}_{AB}]$ are separable channels. Hence, it follows that $R_{s}(\Theta_{AB\to A'B'}[\mathcal{N}_{AB}])\le R_{s}(\mathcal{N}_{AB})$. The analogous inequality for the generalized robustness follows along the same lines.
\end{proof}

In order to calculate the above quantities for the dynamic entanglement resource, i.e., the $K$-swap channel, we review previous results on the robustness of bipartite channels, and especially of unitary bipartite channels:

\begin{theorem}[{\cite[Theorem~5]{nielsen2003QuantumDynamicsPhysical}}]
	Let  $ U_{A_{0}B_{0}} $ be a unitary bipartite operator whose operator Schmidt decomposition reads
	\begin{equation}\label{key}
		U_{A_{0}B_{0}} = \sum_{j} u_{j} A_{j}\otimes B_{j},
	\end{equation}
	where $ A_{j}A_{j}^{\dagger} = \frac{I_{A}}{\abs{A}} $, $ B_{j}B_{j}^{\dagger} = \frac{I_{B}}{\abs{B}} $, and $ \tr A_{j}^{\dagger}A_{k} = \tr B_{j}^{\dagger} B_{k} = \delta_{jk} $. Then its robustness is given by
	\begin{equation}\label{key}
		R_{s}(\mathcal{U}_{AB}) = R(\mathcal{U}_{AB}) = \dfrac{\left( \sum_{j}u_{j} \right)^{2}}{\abs{A} \abs{B}} - 1.
	\end{equation}
\end{theorem}
The swap operator $ F_{AB}^{K} $ acting on $ K $-dimensional subsystems can be written as
\begin{equation}\label{key}
	F_{AB}^{K} = \sum_{i=1}^{K^{2}} G_{i}\otimes G_{i}^{\dagger}
\end{equation}
for any orthonormal operator basis $ \{G_{i}\}_{i=1}^{K^{2}} $ such that $ \tr G_{i}^{\dag}G_{j} = \delta_{ij} $ \cite{wolf2012QuantumChannelsOperations}. Using an orthonormal unitary basis, e.g., the discrete Weyl basis\footnote{The discrete Weyl basis is composed by the $d^2$ unitary operators $ U_{kl} = \sum_{s = 0}^{d - 1}e^{\frac{2\pi i}{d}s l}\ket{k+s}\!\bra{s}$, where $k, l = 0, 1,\dots, d - 1$.}, the operator Schmidt decomposition of the swap gate is given by
\begin{equation}\label{key}
	F_{AB}^{K} = \sum_{i = 1}^{K^{2}} \dfrac{U_{i}}{\sqrt{K}}\otimes \dfrac{U_{i}^{\dag}}{\sqrt{K}}.
\end{equation}
Hence, the robustness of the $K$-swap channel is given as follows:
\begin{equation}\label{key}
	R_{s}(\mathcal{F}_{AB}^{K}) = R(\mathcal{F}_{AB}^{K}) = K^{2} - 1.
\end{equation}
The following well-known fact concerning separability of the isotropic states will be used afterwards:
\begin{theorem}[{\cite{horodecki1999reduction}}]\label{thm: isotropic state separabilitiy}
	The isotropic state
	\begin{equation}\label{key}
		p\Phi_{A_{0}B_{0}}^{K} + (1-p) \dfrac{I_{A_{0}B_{0}} - \Phi_{A_{0}B_{0}}^{K}}{K^{2}-1}\quad ( 0\le p \le 1 )
	\end{equation}
	is separable if and only if $ p\le \dfrac{1}{K} $.
\end{theorem}

\section{One-Shot Dynamic Entanglement Cost of a Bipartite Quantum Channel}

\begin{figure}[tbph]
	\centering	
	\tikzset{every picture/.style={line width=0.75pt}} 
	\begin{tikzpicture}[x=0.75pt,y=0.75pt,yscale=-0.7,xscale=0.7]
	
	\draw   (142,236) -- (213.5,236) -- (213.5,300) -- (142,300) -- cycle ;
	\draw    (33.5,215) -- (61.5,215) ;
	\draw    (33.5,298) -- (61.5,298) ;
	\draw    (290.5,297) -- (315.5,297) ;
	\draw    (290.5,214) -- (315.5,214) ;
	\draw   (289.5,315) -- (253.5,314) -- (253.5,220) -- (97.5,218) -- (97.5,318) -- (61.5,318) -- (61.5,175) -- (289.5,174) -- cycle ;
	\draw    (98.5,247) -- (142.5,247) ;
	\draw    (98.5,288) -- (142.5,288) ;
	\draw    (213.5,247) -- (253.5,247) ;
	\draw    (213.5,288) -- (253.5,288) ;
	\draw   (472,195) -- (542,195) -- (542,319) -- (472,319) -- cycle ;
	\draw    (541.5,215) -- (580.5,215) ;
	\draw    (541.5,300) -- (580.5,300) ;
	\draw    (432.5,215) -- (471.5,215) ;
	\draw    (431.5,300) -- (470.5,300) ;
	
	\draw (177.75,268) node   [align=left] {$ \mathcal{F}_{A'B'}^{K} $};
	\draw (16,214) node   [align=left] {$ A_{0} $};
	\draw (16,298) node   [align=left] {$ B_{0} $};
	\draw (332,213) node   [align=left] {$ A_{1} $};
	\draw (331,298) node   [align=left] {$ B_{1} $};
	\draw (121,236) node   [align=left] {$ A_{0}' $};
	\draw (122,275) node   [align=left] {$ B_{0}' $};
	\draw (233,235) node   [align=left] {$ A_{1}' $};
	\draw (234,273) node   [align=left] {$ B_{1}' $};
	\draw (171,193) node   [align=left] {$\Theta_{A'B'\to AB} $};
	\draw (507,257) node   [align=left] {$ \mathcal{N}_{AB} $};
	\draw (414,214) node   [align=left] {$ A_{0} $};
	\draw (413,298) node   [align=left] {$ B_{0} $};
	\draw (602,214) node   [align=left] {$ A_{1} $};
	\draw (603,300) node   [align=left] {$ B_{1} $};
	\draw (372,249) node    {$\approx_{\varepsilon} $};
	
	\end{tikzpicture}
	\caption{The one-shot dynamic entanglement cost of a bipartite quantum channel $ \mathcal{N}_{AB} $ under a free superchannel $ \Theta_{A'B'\to AB} $.}
	\label{fig:entanglement_cost}
\end{figure}
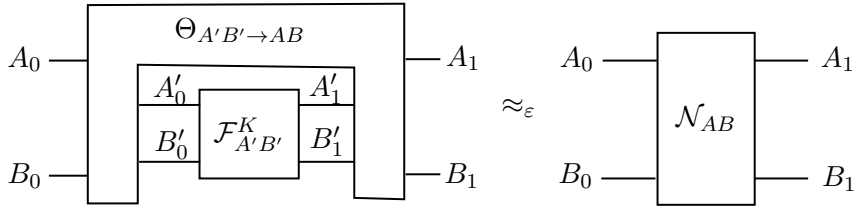

The first operational task we investigate consists in simulating a single instance of a known channel $\mathcal{N}_{AB}$ using a $K$-swap channel --- with $K$ as small as possible --- together with free superchannels, as depicted in Fig.~\ref{fig:entanglement_cost}. One might call this task dynamic entanglement dilution, in analogy to the entanglement dilution task for quantum states. We can thus give the following formal definition.

\begin{definition}
	Given $ \varepsilon\ge 0 $, the one-shot dynamic entanglement cost of a bipartite quantum channel $ \mathcal{N}_{AB} $ under SEPPSC is defined as follows:
	\begin{align*}
	E_{C,\seppsc}^{(1), \varepsilon} (\mathcal{N}_{AB})\coloneqq \min &\left\{\log K^{2} : \dfrac{1}{2}\dnorm{\Theta_{A'B'\to AB}[\mathcal{F}_{A'B'}^{K}] - \mathcal{N}_{AB}} \le \varepsilon,\right .\\
	&\left . \quad\Theta_{A'B'\to AB}\in \seppsc(A'\!:\!B'\to A\!:\!B),\ K\in \mathbb{N}_{0}\right\}.
	\end{align*}
\end{definition}

We now present our first main result. It is a two-fold bound that connects the one-shot dynamic entanglement cost with the smooth standard log-robustness, thus, providing an operational meaning of the latter quantity. \\


\begin{theorem}\label{thm: one-shot dyn. ent. cost}
	Given $ \varepsilon\ge 0 $, the one-shot dynamic entanglement cost of a bipartite quantum channel $ \mathcal{N}_{AB} $ under SEPPSC is bounded as
	\begin{equation}\label{key}
	LR_{s}^{\varepsilon}(\mathcal{N}_{AB}) \le E_{C,\seppsc}^{(1), \varepsilon} (\mathcal{N}_{AB}) \le LR_{s}^{\varepsilon}(\mathcal{N}_{AB}) + 2.
	\end{equation}
\end{theorem}

\begin{proof} We break down the argument into separate proofs of the two bounds.
\begin{enumerate}[(i)]
    \item For the lower bound, let $ \Theta_{A'B'\to AB} \in \seppsc(A':B'\to A:B) $ be a superchannel that achieves $ E_{C,\seppsc}^{(1), \varepsilon} (\mathcal{N}_{AB}) $ with $ \mathcal{F}_{AB}^{K} $, that is, $ \Theta_{A'B'\to AB} [ \mathcal{F}_{A'B'}^{K}] \approx_{\varepsilon} \mathcal{N}_{AB}$. Then we have that
	\begin{align*}
	LR_{s}^{\varepsilon} (\mathcal{N}_{AB}) &\le LR_{s}\left( \Theta_{A'B'\to AB}[\mathcal{F}_{A'B'}^{K}] \right)\\
	&\le LR_{s}(\mathcal{F}_{A'B'}^{K})=\log K^{2} = E_{C,\seppsc}^{(1), \varepsilon} (\mathcal{N}_{AB}). 
	\end{align*}
	
	\item For the upper bound, let $ \mathcal{N}_{AB}^{\varepsilon} $ be a channel such that
	\begin{equation}\label{key}
	LR_{s}^{\varepsilon} (\mathcal{N}_{AB}) = LR_{s}(\mathcal{N}_{AB}^{\varepsilon}) =  \log (1+r),
	\end{equation}
	where $ r = R_{s} (\mathcal{N}_{AB}^{\varepsilon})$. There exists a separable channel $ \mathcal{M}_{AB} $ such that
	\begin{equation}\label{key}
	\dfrac{\mathcal{N}_{AB}^{\varepsilon} + r \mathcal{M}_{AB}}{1+r} \in \sepc(A:B).
	\end{equation}
	Let $ J_{A_{0}B_{0}\widetilde{A}_{1}\widetilde{B}_{1}}^{\mathcal{E}_{AB}} \coloneqq \mathsf{id}_{AB}\otimes \mathcal{E}_{\widetilde{A}\widetilde{B}}\left( \Phi_{A_{0}\widetilde{A}_{0}}^{K} \otimes \Phi_{B_{0}\widetilde{B}_{0}}^{K}\right)$ be the (normalized) Choi matrix for a quantum channel $ \mathcal{E}_{AB} $, where $ \ket{\Phi^{K}}_{A_{0}B_{0}} = \frac{1}{\sqrt{K}}\sum_{i=0}^{K-1}\ket{ii}_{A_{0}B_{0}} $ is the maximally entangled state. Setting $ K = \lceil \sqrt{1+r} \rceil $, we construct a SEPPSC $ \Theta_{A'B'\to AB} $ that simulates $ \mathcal{N}_{AB}^{\varepsilon} $, that is, $ \Theta_{A'B'\to AB}[\mathcal{F}_{A'B'}^{K}]  = \mathcal{N}_{AB}^{\varepsilon}$ as follows:
	\begin{align*}\label{key}
		\Theta_{A'B'\to AB}[\mathcal{E}_{A'B'}] =& \tr \left( J_{A_{0}B_{0}\widetilde{A}_{1}\widetilde{B}_{1}}^{\mathcal{F}_{A'B'}^{K}} J_{A_{0}B_{0}\widetilde{A}_{1}\widetilde{B}_{1}}^{\mathcal{E}_{A'B'}} \right)  \mathcal{N}_{AB}^{\varepsilon}\\
		& + \tr \left\{ \left(I_{A_{0}B_{0}\widetilde{A}_{1}\widetilde{B}_{1}} - J_{A_{0}B_{0}\widetilde{A}_{1}\widetilde{B}_{1}}^{\mathcal{F}_{A'B'}^{K}} \right) J_{A_{0}B_{0}\widetilde{A}_{1}\widetilde{B}_{1}}^{\mathcal{E}_{A'B'}} \right\} \mathcal{M}_{AB}.
	\end{align*}
	While $ \Theta_{A'B'\to AB}[\mathcal{F}_{A'B'}^{K}]  = \mathcal{N}_{AB}^{\varepsilon}$ is apparent from the trace terms, we show that $ \Theta_{A'B'\to AB} $ is a SEPPSC. Note that the Choi matrix of a separable channel $ \mathcal{E}_{A'B'} $ is a separable state and $ J_{A_{0}B_{0}\widetilde{A}_{1}\widetilde{B}_{1}}^{\mathcal{F}_{A'B'}^{K}} = \Phi_{A_{0}\widetilde{B}_{1}}^{K} \otimes \Phi_{\widetilde{A}_{1}B_{0}}^{K} $, which leads to
	\begin{equation}\label{key}
	\tr \left( J_{A_{0}B_{0}\widetilde{A}_{1}\widetilde{B}_{1}}^{\mathcal{F}_{A'B'}^{K}} J_{A_{0}B_{0}\widetilde{A}_{1}\widetilde{B}_{1}}^{\mathcal{E}_{A'B'}} \right) \le \dfrac{1}{K^{2}}
	\end{equation}
	for any $ \mathcal{E}_{A'B'}\in \sepc(A':B') $.\footnote{This follows directly from a well-known result that is reported in the Appendix as Proposition \ref{thm: max fidelity btw. ent. sep. states}.} Therefore, when $ \mathcal{E}_{A'B'} \in \sepc(A':B') $, we have that
	\begin{align*}
	\Theta_{A'B'\to AB}[\mathcal{E}_{A'B'}] &=\tr \left( J_{A_{0}B_{0}\widetilde{A}_{1}\widetilde{B}_{1}}^{\mathcal{F}_{A'B'}^{K}} J_{A_{0}B_{0}\widetilde{A}_{1}\widetilde{B}_{1}}^{\mathcal{E}_{A'B'}} \right)  \mathcal{N}_{AB}^{\varepsilon}\\
	&\quad + \tr \left\{ \left(I_{A_{0}B_{0}\widetilde{A}_{1}\widetilde{B}_{1}} - J_{A_{0}B_{0}\widetilde{A}_{1}\widetilde{B}_{1}}^{\mathcal{F}_{A'B'}^{K}} \right) J_{A_{0}B_{0}\widetilde{A}_{1}\widetilde{B}_{1}}^{\mathcal{E}_{A'B'}} \right\} \mathcal{M}_{AB}\\
	&= q' \mathcal{N}_{AB}^{\varepsilon} + (1-q') \mathcal{M}_{AB}\\
	&= q\left( \dfrac{\mathcal{N}_{AB}^{\varepsilon} + r \mathcal{M}_{AB}}{1+r} \right) + (1-q) \mathcal{M}_{AB} 
	\in \sepc(A:B),
	\end{align*}
	where $ q= q' (1+r) \le 1 $ due to $ q'\le \frac{1}{K^{2}}=\frac{1}{\lceil \sqrt{1+r} \rceil^{2}} $. We conclude that 
	\begin{align*}
	E_{C,\seppsc}^{(1), \varepsilon} (\mathcal{N}_{AB}) &\le \log K^{2}\\
	&= 2 \log \lceil \sqrt{1+r}\rceil\\
	&\le 2 \log (2\sqrt{1+r})\\
    &= \log (1+r) + 2 \\
	&\le LR_{s}^{\varepsilon}(\mathcal{N}_{AB}) + 2,
	\end{align*}
    where in the third line we observed that $\lceil x \rceil \leq 2x$ for all $x\geq 1$. This concludes the proof.
\end{enumerate}
\end{proof}

\section{One-Shot Distillable Dynamic Entanglement of a Bipartite Quantum Channel}

\begin{figure}[tbhp]
	\centering	
	\tikzset{every picture/.style={line width=0.75pt}} 
	\begin{tikzpicture}[x=0.75pt,y=0.75pt,yscale=-0.7,xscale=0.7]
	
	\draw   (142,236) -- (213.5,236) -- (213.5,300) -- (142,300) -- cycle ;
	\draw    (33.5,215) -- (61.5,215) ;
	\draw    (33.5,298) -- (61.5,298) ;
	\draw    (290.5,297) -- (315.5,297) ;
	\draw    (290.5,214) -- (315.5,214) ;
	\draw   (289.5,315) -- (253.5,314) -- (253.5,220) -- (97.5,218) -- (97.5,318) -- (61.5,318) -- (61.5,175) -- (289.5,174) -- cycle ;
	\draw    (98.5,247) -- (142.5,247) ;
	\draw    (98.5,288) -- (142.5,288) ;
	\draw    (213.5,247) -- (253.5,247) ;
	\draw    (213.5,288) -- (253.5,288) ;
	\draw   (472,195) -- (542,195) -- (542,319) -- (472,319) -- cycle ;
	\draw    (541.5,215) -- (580.5,215) ;
	\draw    (541.5,300) -- (580.5,300) ;
	\draw    (432.5,215) -- (471.5,215) ;
	\draw    (431.5,300) -- (470.5,300) ;
	
	\draw (177.75,268) node   [align=left] {$ \mathcal{N}_{AB} $};
	\draw (16,214) node   [align=left] {$ A_{0}' $};
	\draw (16,298) node   [align=left] {$ B_{0}' $};
	\draw (332,213) node   [align=left] {$ A_{1}' $};
	\draw (331,298) node   [align=left] {$ B_{1}' $};
	\draw (121,236) node   [align=left] {$ A_{0} $};
	\draw (122,275) node   [align=left] {$ B_{0} $};
	\draw (233,235) node   [align=left] {$ A_{1} $};
	\draw (234,273) node   [align=left] {$ B_{1} $};
	\draw (171,193) node   [align=left] {$\Theta_{AB\to A'B'} $};
	\draw (507,257) node   [align=left] {$ \mathcal{F}_{A'B'}^{K} $};
	\draw (414,214) node   [align=left] {$ A_{0}' $};
	\draw (413,298) node   [align=left] {$ B_{0}' $};
	\draw (602,214) node   [align=left] {$ A_{1}' $};
	\draw (603,300) node   [align=left] {$ B_{1}' $};
	\draw (372,249) node    {$\approx_{\varepsilon} $};
	
	\end{tikzpicture}
	\caption{The one-shot distillable dynamic entanglement of a bipartite quantum channel $ \mathcal{N}_{AB} $ under a free superchannel $ \Theta_{AB\to A'B'} $.}
	\label{fig:distillable_entanglement}
\end{figure}
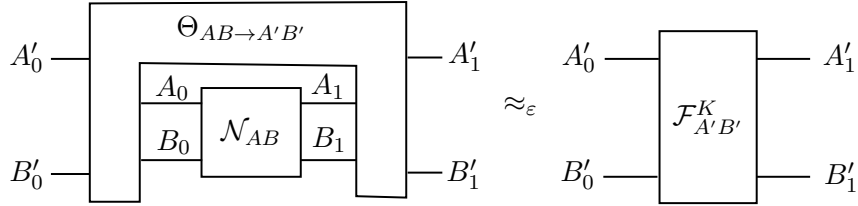

The converse task to dynamic entanglement dilution is dynamic entanglement distillation. In our setting, this can be thought of as the task of simulating a $K$-swap channel --- with $K$ as large as possible --- using a noisy channel as a dynamic entanglement resource together with free superchannels. We give a pictorial representation of the process in Fig.~\ref{fig:distillable_entanglement}. We can capture this notion through the following formal definition.

\begin{definition} \label{1-shot_distillable_def}
	Given $ \varepsilon\ge 0 $, the one-shot distillable dynamic entanglement of a bipartite quantum channel $ \mathcal{N}_{AB} $ under SEPPSC is defined as
	\begin{align*}\label{key}
	E_{D, \seppsc}^{(1), \varepsilon} (\mathcal{N}_{AB}) \coloneqq \max &\left\{ \log K^{2} : \dfrac{1}{2}\pnorm[\diamond]{\Theta_{AB\to A'B'}[\mathcal{N}_{AB}] - \mathcal{F}_{A'B'}^{K}}\le \varepsilon,\right .\\
	&\left . \quad \Theta_{AB\to A'B'} \in \seppsc(A\!:\!B\to A'\!:\!B'),\ K\in \mathbb{N}_{0} \right\}.
	\end{align*}
\end{definition}

We propose to bound the above operational quantity with a measure that is inspired by the one-shot distillable entanglement of a quantum state \cite{brandao2011OneShotRatesEntanglement}.
It is obtained from the hypothesis-testing relative entropy of channels by means of an additional minimization over the set of separable channels \cite{cooney2016StrongConverseExponents, xyuan2019HypothesisTestingEntropies}:
\begin{definition} \label{hypothesis_testing_def}
	Given $ \varepsilon\ge 0 $, we define the hypothesis-testing relative entropy of dynamic entanglement of a bipartite quantum channel $ \mathcal{N}_{AB}$ by
	\begin{align*}
	E_{H}^{\varepsilon}(\mathcal{N}_{AB}) \coloneqq &  \max_{\Psi_{A_{0}B_{0}R_{0}}}  \sup_{\substack{0\le Q_{A_{1}B_{1}R_{0}} \le I_{A_{1}B_{1}R_{0}}\\ \tr \left\{ Q_{A_{1}B_{1}R_{0}}\cdot \mathcal{N}_{AB}\otimes \mathsf{id}_{R_{0}}(\Psi_{A_{0}B_{0}R_{0}}) \right\}\ge 1-\varepsilon}}\\
	& \min_{\mathcal{M}_{AB}\in \sepc(A:B)} \left\{ -\log \tr \left( Q_{A_{1}B_{1}R_{0}} \cdot\mathcal{M}_{AB}\otimes \mathsf{id}_{R_{0}}(\Psi_{A_{0}B_{0}R_{0}}) \right) \right\}.
	\end{align*}
\end{definition}

We remark that the above quantity is monotonic in $\varepsilon$ from the definition, implying in particular that $ E_{H}^{\varepsilon}(\mathcal{N}_{AB}) \ge E_{H}^{\varepsilon/2}(\mathcal{N}_{AB}) $. Moreover, the hypothesis-testing relative entropy of dynamic entanglement does not increase under SEPPSC:
\begin{proposition}
	For a bipartite quantum channel $ \mathcal{N}_{AB} $, and $ \varepsilon\ge 0 $, it holds that
	\begin{equation}\label{key}
	E_{H}^{\varepsilon}(\Theta_{AB\to A'B'}[\mathcal{N}_{AB}]) \le E_{H}^{\varepsilon}(\mathcal{N}_{AB}) \quad \forall\ \Theta_{AB\to A'B'} \in \seppsc(AB\to A'B').
	\end{equation}
\end{proposition}
\begin{proof}
	Let $\Psi_{A_{0}'B_{0}'R_{0}}^{\ast}$ and $Q_{A_{1}'B_{1}'R_{0}}^{\ast}$ be optimal arguments of $E_{H}^{\varepsilon}(\Theta_{AB\to A'B'}[\mathcal{N}_{AB}])$, so that
	\begin{equation}
		E_{H}^{\varepsilon}(\Theta_{AB\to A'B'}[\mathcal{N}_{AB}]) = \min_{\mathcal{M}_{A'B'}\in \sepc(A':B')}\\
		 \left \{ -\log \tr Q_{A_{1}'B_{1}'R_{0}}^{\ast} \cdot \mathcal{M}_{A'B'}\otimes \mathsf{id}_{R_{0}}\left (\Psi_{A_{0}'B_{0}'R_{0}}^{\ast}\right ) \right \} ,
	\end{equation}
	where $ 0\le Q_{A_{1}'B_{1}'R_{0}}^{\ast} \le I_{A_{1}'B_{1}'R_{0}} $ and
	\[
	\tr \left\{ Q_{A_{1}'B_{1}'R_{0}}^{\ast}\cdot \Theta_{AB\to A'B'}[\mathcal{N}_{AB}]\otimes \mathsf{id}_{R_{0}}\left (\Psi_{A_{0}'B_{0}'R_{0}}^{\ast}\right ) \right\}\ge 1-\varepsilon.\] Then using the structure of the superchannel $ \Theta_{AB\to A'B'}[\mathcal{E}_{AB}] = \mathcal{U}_{A_{1}B_{1}E_{0}\to A_{1}'B_{1}'} \circ \mathcal{E}_{AB}\circ \mathcal{W}_{A_{0}'B_{0}'\to A_{0}B_{0}E_{0}}$ with isometries $ \mathcal{U}_{A_{1}'B_{1}'\to A_{1}B_{1}E_{0}}^{\dag} $ and $ \mathcal{W}_{A_{0}'B_{0}'\to A_{0}B_{0}E_{0}} $, we observe that
	\begin{align*}
	E_{H}^{\varepsilon}(\Theta_{AB\to A'B'}[\mathcal{N}_{AB}]) & = \min_{\mathcal{M}_{A'B'}\in \sepc(A':B')} \left \{ -\log \tr Q_{A_{1}'B_{1}'R_{0}}^{\ast} \cdot \mathcal{M}_{A'B'}\otimes \mathsf{id}_{R_{0}}\left (\Psi_{A_{0}'B_{0}'R_{0}}^{\ast}\right ) \right\}\\
	& \le \min_{\mathcal{\widetilde{M}}_{AB}\in \sepc(A:B)} \left \{  -\log \tr Q_{A_{1}'B_{1}'R_{0}}^{\ast} \cdot \Theta_{AB\to A'B'}\left [\mathcal{\widetilde{M}}_{AB}\right ]\otimes \mathsf{id}_{R_{0}}\left (\Psi_{A_{0}'B_{0}'R_{0}}^{\ast}\right ) \right\}\\
	& = \min_{\mathcal{\widetilde{M}}_{AB}\in \sepc(A:B)} \left [ -\log \tr \left \{ \mathcal{U}_{A_{1}'B_{1}'\to  A_{1}B_{1}E_{0}}^{\dag}\left (Q_{A_{1}'B_{1}'R_{0}}^{\ast} \right ) \cdot  \right . \right .\\
	&\qquad\qquad \qquad\qquad \qquad\qquad \qquad \left .\left .\mathcal{\widetilde{M}}_{AB}\otimes \mathsf{id}_{E_{0}R_{0}}\left (\mathcal{W}_{A_{0}'B_{0}'\to A_{0}B_{0}E_{0}}\left (\Psi_{A_{0}'B_{0}'R_{0}}^{\ast}\right )\right ) \right \} \right]\\
	&= \min_{\mathcal{\widetilde{M}}_{AB}\in \sepc(A:B)} \left \{ -\log \tr \widetilde{Q}_{A_{1}B_{1}E_{0}R_{0}}^{\ast} \cdot  \mathcal{\widetilde{M}}_{AB}\otimes \mathsf{id}_{E_{0}R_{0}}\left (\widetilde{\Psi}_{A_{0}B_{0}E_{0}R_{0}}^{\ast}\right ) \right\}\\
	& \le \max_{\Psi_{A_{0}B_{0}E_{0}R_{0}}}  \sup_{\substack{0\le Q_{A_{1}B_{1}E_{0}R_{0}} \le I_{A_{1}B_{1}E_{0}R_{0}}\\ \tr \left\{ Q_{A_{1}B_{1}E_{0}R_{0}}\cdot \mathcal{N}_{AB}\otimes \mathsf{id}_{E_{0}R_{0}}(\Psi_{A_{0}B_{0}E_{0}R_{0}}) \right\}\ge 1-\varepsilon}} \min_{\mathcal{\widetilde{M}}_{AB}\in \sepc(A:B)} \\
	&\qquad \qquad \left\{ -\log \tr \left( Q_{A_{1}B_{1}E_{0}R_{0}} \cdot\mathcal{\widetilde{M}}_{AB}\otimes \mathsf{id}_{E_{0}R_{0}}\left (\Psi_{A_{0}B_{0}E_{0}R_{0}}\right ) \right) \right\}\\
	&= E_{H}^{\varepsilon}(\mathcal{N}_{AB}),
	\end{align*}
	where $ \widetilde{Q}_{A_{1}B_{1}E_{0}R_{0}}^{\ast}=  \mathcal{U}_{A_{1}'B_{1}'\to A_{1}B_{1}E_{0}}^{\dag}(Q_{A_{1}'B_{1}'R_{0}}^{\ast} )$ and $ \widetilde{\Psi}_{A_{0}B_{0}E_{0}R_{0}}^{\ast}= \mathcal{W}_{A_{0}'B_{0}'\to A_{0}B_{0}E_{0}}(\Psi_{A_{0}'B_{0}'R_{0}}^{\ast})$. The last inequality holds since $ 0\le \widetilde{Q}_{A_{1}B_{1}E_{0}R_{0}}^{\ast}\le I_{A_{1}B_{1}E_{0}R_{0}} $ and
	\begin{equation}\label{key}
	\tr \left\{ \widetilde{Q}_{A_{1}B_{1}E_{0}R_{0}}^{\ast}\cdot \mathcal{N}_{AB}\otimes \mathsf{id}_{E_{0}R_{0}}\left (\widetilde{\Psi}_{A_{0}B_{0}E_{0}R_{0}}^{\ast}\right ) \right\}\ge 1-\varepsilon.
	\end{equation}
	This completes the proof.
\end{proof}

Our second main result connects the two notions identified in Definitions~\ref{1-shot_distillable_def} and~\ref{hypothesis_testing_def}.

\begin{theorem}
	Given $\varepsilon\ge 0$ and a bipartite quantum channel $ \mathcal{N}_{AB} $, if $\left \lfloor E_{H}^{\varepsilon} (\mathcal{N}_{AB}) \right \rfloor$ is even, the one-shot distillable dynamic entanglement from a bipartite quantum channel $\mathcal{N}_{AB}$ under SEPPSC is bounded as
	\begin{equation}\label{key}
	\left \lfloor E_{H}^{\varepsilon} (\mathcal{N}_{AB}) \right \rfloor \le E_{D, \seppsc}^{(1), \varepsilon} (\mathcal{N}_{AB}) \le E_{H}^{2 \varepsilon}(\mathcal{N}_{AB}).
	\end{equation}
	If $\left \lfloor E_{H}^{\varepsilon} (\mathcal{N}_{AB}) \right \rfloor$ is odd, then we have instead that
	\begin{equation}\label{key}
	E_{H}^{\varepsilon} (\mathcal{N}_{AB}) -1 \le E_{D, \seppsc}^{(1), \varepsilon} (\mathcal{N}_{AB}) \le E_{H}^{2 \varepsilon}(\mathcal{N}_{AB}).
	\end{equation}
\end{theorem}
\begin{proof}
We break down the argument into separate proofs of the two inequalities.
\begin{enumerate}[(i)]
\item For the upper bound, let $ \Theta_{AB\to A'B'} $ be an optimal SEPPSC such that $ \Theta_{AB\to A'B'}[\mathcal{N}_{AB}]\approx_{\varepsilon}\mathcal{F}_{A'B'}^{K} $ with $ E_{D, \seppsc}^{(1), \varepsilon} (\mathcal{N}_{AB}) =  \log K^{2}$. From the above two propositions we have that
\begin{align*}
E_{H}^{2 \varepsilon}(\mathcal{N}_{AB})&\ge E_{H}^{2 \varepsilon}(\Theta_{AB\to A'B'}[\mathcal{N}_{AB}]) \\
& = \max_{\Psi_{A_{0}'B_{0}'R_{0}}}  \sup_{\substack{0\le Q_{A_{1}'B_{1}'R_{0}} \le I_{A_{1}'B_{1}'R_{0}}\\ \tr \left\{ Q_{A_{1}'B_{1}'R_{0}}\cdot \Theta_{AB\to A'B'}[\mathcal{N}_{AB}]\otimes \mathsf{id}_{R_{0}}\left (\Psi_{A_{0}'B_{0}'R_{0}}\right ) \right\}\ge 1-2 \varepsilon}} \\
&\qquad  \min_{\mathcal{M}_{A'B'}\in \sepc(A':B')} \left[ -\log \tr \left\{ Q_{A_{1}'B_{1}'R_{0}} \cdot\mathcal{M}_{A'B'}\otimes \mathsf{id}_{R_{0}}\left (\Psi_{A_{0}'B_{0}'R_{0}}\right ) \right\} \right]\\
&\ge \min_{\mathcal{M}_{A'B'}\in \sepc(A':B')} \left \{ -\log \tr \mathcal{F}_{A'B'}^{K}\otimes \mathsf{id}_{\widetilde{A}_{0}'\widetilde{B}_{0}'}\left( \Phi_{A_{0}'\widetilde{A}_{0}'}^{K}\otimes \Phi_{B_{0}'\widetilde{B}_{0}'}^{K} \right) \right .\\
&\qquad \left .\cdot\mathcal{M}_{A'B'}\otimes \mathsf{id}_{\widetilde{A}_{0}'\widetilde{B}_{0}'}\left( \Phi_{A_{0}'\widetilde{A}_{0}'}^{K}\otimes \Phi_{B_{0}'\widetilde{B}_{0}'}^{K} \right)\right \}\\
&= \log K^{2}\\
&= E_{D, \seppsc}^{(1), \varepsilon} (\mathcal{N}_{AB}),
\end{align*}
where the second inequality has been derived by making the ansatz $\Psi_{A_{0}'B_{0}'R_{0}} = \Phi_{A_{0}'\widetilde{A}_{0}'}^{K}\otimes \Phi_{B_{0}'\widetilde{B}_{0}'}^{K}$, and the fourth line follows from Proposition~\ref{thm: max fidelity btw. ent. sep. states}. That this is a valid choice is confirmed by the fact that
\begin{align*}
\dfrac{1}{2}\pnorm[1]{\Theta_{AB\to A'B'}[\mathcal{N}_{AB}](\Psi_{A_{0}'B_{0}'R_{0}}) - \mathcal{F}_{A'B'}^{K}(\Psi_{A_{0}'B_{0}'R_{0}})} \le \dfrac{1}{2}\pnorm[\diamond]{\Theta_{AB\to A'B'}[\mathcal{N}_{AB}] - \mathcal{F}_{A'B'}^{K}} \le \varepsilon
\end{align*}
for any (pure) state $ \Psi_{A_{0}'B_{0}'R_{0}} $, in turn implying that\footnote{Here we are making use of the Fuchs--van de Graaf inequalities~\cite{Fuchs1999}. They establish the relations $1-\sqrt{F(\rho,\sigma)} \le \dfrac{1}{2} \pnorm[1]{\rho - \sigma} \le \sqrt{1- F(\rho, \sigma)}$ between trace distance and quantum fidelity.}
\begin{align*}\label{key}
&F\left (\Theta_{AB\to A'B'}[\mathcal{N}_{AB}](\Psi_{A_{0}'B_{0}'R_{0}}), \mathcal{F}_{A'B'}^{K}(\Psi_{A_{0}'B_{0}'R_{0}})\right )\\
&\ge \left( 1 - \dfrac{1}{2}\pnorm[1]{\Theta_{AB\to A'B'}[\mathcal{N}_{AB}](\Psi_{A_{0}'B_{0}'R_{0}}) - \mathcal{F}_{A'B'}^{K}(\Psi_{A_{0}'B_{0}'R_{0}})} \right)^{2}\\
&\ge \left (1-\varepsilon\right )^{2}\\
&\ge 1 - 2\varepsilon.
\end{align*}

\item For the lower bound, let $ \Psi_{A_{0}B_{0}R_{0}}^{\ast} $ and $ Q_{A_{1}B_{1}R_{0}}^{\ast} $ be optimal arguments of $ E_{H}^{\varepsilon}(\mathcal{N}_{AB}) $, which satisfy that
\begin{align*}
2^{-E_{H}^{\varepsilon}(\mathcal{N}_{AB})} &= \max_{\mathcal{M}_{AB}\in \sepc(A:B)} \left\{ \tr \left( Q_{A_{1}B_{1}R_{0}}^{\ast} \cdot\mathcal{M}_{AB}\otimes \mathsf{id}_{R_{0}}\left (\Psi_{A_{0}B_{0}R_{0}}^{\ast}\right ) \right) \right\}.
\end{align*}

Setting $ K=2^{\frac{1}{2} \left \lfloor E_{H}^{\varepsilon}(\mathcal{N}_{AB}) \right \rfloor} $ for $ \left \lfloor E_{H}^{\varepsilon} (\mathcal{N}_{AB}) \right \rfloor $ even, and $ K=2^{\left \lfloor \frac{1}{2}E_{H}^{\varepsilon}(\mathcal{N}_{AB}) \right \rfloor} $ otherwise, we can construct a SEPPSC $\Theta_{AB\to A'B'}$ as follows:
\begin{align*}
\Theta_{AB\to A'B'}[\mathcal{E}_{AB}] \coloneqq &\tr \left \{Q_{A_{1}B_{1}R_{0}}^{\ast}\mathcal{E}_{AB}(\Psi_{A_{0}B_{0}R_{0}}^{\ast})\right \} \mathcal{F}_{A'B'}^{K} \\
&+ \tr \left \{\left( I_{A_{1}B_{1}R_{0}} - Q_{A_{1}B_{1}R_{0}}^{\ast} \right) \mathcal{E}_{AB}(\Psi_{A_{0}B_{0}R_{0}}^{\ast}) \right \} \mathcal{G}_{A'B'}^{K},
\end{align*}
where $\mathcal{G}_{A'B'}^{K}$ is the quantum channel corresponding to the following (normalized) Choi operator:
\begin{equation}
J_{A_{0}'B_{0}'\widetilde{A}_{1}'\widetilde{B}_{1}'}^{\mathcal{G}_{A'B'}^{K}} = \dfrac{I_{A_{0}'B_{0}'\widetilde{A}_{1}'\widetilde{B}_{1}'} - J_{A_{0}'B_{0}'\widetilde{A}_{1}'\widetilde{B}_{1}'}^{\mathcal{F}_{A'B'}^{K}}}{K^{4}-1}=\dfrac{I_{A_{0}'B_{0}'\widetilde{A}_{1}'\widetilde{B}_{1}'}-\Phi_{A_{0}'\widetilde{B}_{1}'}^{K} \otimes \Phi_{\widetilde{A}_{1}'B_{0}'}^{K}}{K^{4}-1}\in \sepd(A_{0}'\widetilde{A}_{1}':B_{0}'\widetilde{B}_{1}'),
\end{equation}
which implies that $ \mathcal{G}_{A'B'}^{K}\in \sepc(A':B') $.
For $ \mathcal{M}_{AB}\in \sepc(A:B) $, we observe that
\begin{align*}
\Theta_{AB\to A'B'}[\mathcal{M}_{AB}] &\coloneqq \tr \left \{Q_{A_{1}B_{1}R_{0}}^{\ast}\mathcal{M}_{AB}(\Psi_{A_{0}B_{0}R_{0}}^{\ast})\right \} \mathcal{F}_{A'B'}^{K} \\
&\quad + \tr \left \{\left( I_{A_{1}B_{1}R_{0}} - Q_{A_{1}B_{1}R_{0}}^{\ast} \right) \mathcal{M}_{AB}(\Psi_{A_{0}B_{0}R_{0}}^{\ast}) \right \} \mathcal{G}_{A'B'}^{K}\\
&= q \mathcal{F}_{A'B'}^{K} + (1-q) \mathcal{G}_{A'B'}^{K}\\
&\in \sepc(A':B')
\end{align*}
because of $ q=\tr \left \{Q_{A_{1}B_{1}R_{0}}^{\ast}\mathcal{M}_{AB}(\Psi_{A_{0}B_{0}R_{0}}^{\ast})\right \}\le \frac{1}{K^{2}} $ and Theorem~\ref{thm: isotropic state separabilitiy} regarding the Choi matrix. Denoting
 \begin{equation}\label{key}
 	q^{\ast} =  \tr \left \{Q_{A_{1}B_{1}R_{0}}^{\ast}\mathcal{N}_{AB}(\Psi_{A_{0}B_{0}R_{0}}^{\ast})\right \} \ge 1-\varepsilon,
 \end{equation}
 we have that
\begin{align*}
\dfrac{1}{2}\pnorm[\diamond]{\Theta_{AB\to A'B'}[\mathcal{N}_{AB}] - \mathcal{F}_{A'B'}^{K}}&= \dfrac{1}{2}\pnorm[\diamond]{q^{\ast} \mathcal{F}_{A'B'}^{K} + (1-q^{\ast}) \mathcal{G}_{A'B'}^{K} - \mathcal{F}_{A'B'}^{K}}\\
&\le \dfrac{1}{2}\pnorm[\diamond]{(1-q^{\ast}) \mathcal{F}_{A'B'}^{K}} + \dfrac{1}{2}\pnorm[\diamond]{(1-q^{\ast}) \mathcal{G}_{A'B'}^{K} }\\
&= 1-q^{\ast}\\
&\le \varepsilon,
\end{align*}
where we used that $ \pnorm[\diamond]{\mathcal{E}} = 1 $ for $ \mathcal{E}\in \cptp(AB) $ \cite{watrous2018TTQI}.
Therefore, we conclude that, for $ \left \lfloor E_{H}^{\varepsilon} (\mathcal{N}_{AB}) \right \rfloor $ even, \begin{equation*}
E_{D, \seppsc}^{(1), \varepsilon} (\mathcal{N}_{AB}) \ge \log K^{2} = \left \lfloor E_{H}^{\varepsilon}(\mathcal{N}_{AB})\right  \rfloor .
\end{equation*}
When $\left \lfloor E_{H}^{\varepsilon} (\mathcal{N}_{AB}) \right \rfloor$ is odd, noticing that it holds that $ \left \lfloor \dfrac{q}{p}\right \rfloor \ge \dfrac{q+1}{p}-1 $ for integers $q,p\in \mathbb{N}$, we obtain that
\begin{align*}
E_{D, \seppsc}^{(1), \varepsilon} (\mathcal{N}_{AB}) &\ge \log K^{2} = 2\left  \lfloor \frac{1}{2}E_{H}^{\varepsilon}(\mathcal{N}_{AB})\right  \rfloor \\
&\ge 2\left( \dfrac{E_{H}^{\varepsilon}(\mathcal{N}_{AB})+1}{2}-1 \right)\\
&= E_{H}^{\varepsilon}(\mathcal{N}_{AB})-1.
\end{align*}
This concludes the proof.
\end{enumerate}
\end{proof}

\section{One-Shot Catalytic Dynamic Entanglement Cost of a Bipartite Quantum Channel}

The third operational task we consider is a variation on the theme of dynamic entanglement cost. We push this notion further by introducing two tweaks: (i)~we allow an additional dynamic entanglement resource that could be used as a catalyst while simulating a bipartite channel, with the stipulation that the catalyst channel be returned intact after the task; and (ii)~we introduce a class of superchannels that might generate a small amount of dynamic entanglement when acting on separable channels.
\begin{definition} \label{delta_seppsc_def}
	For $ \delta\ge 0 $, a superchannel $ \Theta_{AB\to A'B'} $ is called $ \delta $-separability-preserving superchannel ($ \delta $-SEPPSC) if
	\begin{equation}\label{delta_seppsc}
	R \left( \Theta_{AB\to A'B'}[ \mathcal{M}_{AB}] \right)\le \delta \quad \forall \mathcal{M}_{AB} \in \sepc(A:B),
	\end{equation}
	where $ R (\mathcal{N}_{AB}) = \min \left\{ s\ge 0 : \mathcal{N}_{AB}\le (1+s) \mathcal{M}_{AB}, \mathcal{M}_{AB}\in \sepc(A:B) \right\}$ is the generalized robustness with respect to the separable channels.
\end{definition}

The choice of the generalized robustness to quantify the maximum amount of entanglement generation allowed in the above definition may seem rather arbitrary. A compelling reason why this is in fact a reasonable and natural choice comes from the study of entanglement theory for states. Indeed, it is known that a condition analogous to \eqref{delta_seppsc} leads to a universally reversible theory of entanglement manipulation~\cite{BrandaoPlenio1, BrandaoPlenio2}. The role of the generalized robustness in this context is quite unique, in the sense that using alternative measures --- such as the trace norm distance from the set of separable states --- is known to trivialize the problem~\cite[Section~V]{BrandaoPlenio2}. In light of this, Definition~\ref{delta_seppsc_def} identifies a good candidate for a useful enlargement of the set of free superchannels.

As expected, the generalized log-robustness and its smooth version
 might increase under a $ \delta $-separability-preserving superchannel as the following results show:
\begin{proposition}
	Let $ \Theta_{AB\to A'B'} $ be a $ \delta $-SEPPSC. For any bipartite quantum channel $ \mathcal{N}_{AB} $, the following holds:
	\begin{equation}\label{key}
	LR(\Theta_{AB\to A'B'}[\mathcal{N}_{AB}] ) \le LR (\mathcal{N}_{AB}) +\log (1+\delta).
	\end{equation}
\end{proposition}
\begin{proof}
	Let $ r \equiv R(\mathcal{N}_{AB}) $ such that
	\begin{equation}\label{key}
	\mathcal{N}_{AB} + r \mathcal{E}_{AB} = (1+r) \mathcal{M}_{AB},
	\end{equation}
	where $ \mathcal{M}_{AB}\in \sepc(A:B) $.  It follows that 
	\begin{equation}\label{key}
	\Theta_{AB\to A'B'}[\mathcal{N}_{AB}] + r \Theta_{AB\to A'B'}[\mathcal{E}_{AB}] = (1+r) \Theta_{AB\to A'B'}[\mathcal{M}_{AB}].
	\end{equation}
	Also, we have that
	\begin{equation}\label{key}
	\Theta_{AB\to A'B'}[\mathcal{M}_{AB}] + r' \mathcal{G}_{A'B'} = (1+r') \mathcal{M}_{A'B'}',
	\end{equation}
	where $ r'\equiv R(\Theta_{AB\to A'B'}[\mathcal{M}_{AB}])\le \delta $, and $ \mathcal{M}_{A'B'}'\in \sepc(A':B') $. From these two equations, it follows that
	\begin{equation}\label{key}
	\Theta_{AB\to A'B'}[\mathcal{N}_{AB}] + r \Theta_{AB\to A'B'}[\mathcal{E}_{AB}] + (1+r) r' \mathcal{G}_{A'B'} = (1+r) (1+r') \mathcal{M}_{A'B'}',
	\end{equation}
	which implies that $ 1+R(\Theta_{AB\to A'B'}[\mathcal{N}_{AB}]) \le (1+r)(1+r') $.
\end{proof}
\begin{lemma}
	For $ \Theta_{AB\to A'B'} \in \delta\mhyphen\seppsc(A:B\to A':B') $, we have that
	\begin{equation}\label{key}
		LR^{\epsilon} (\Theta_{AB\to A'B'}[\mathcal{N}_{AB}]) \le LR^{\epsilon} (\mathcal{N}_{AB}) + \log (1+\delta).
	\end{equation}
\end{lemma}
\begin{proof}
	Let $ \mathcal{N}_{AB}^{\varepsilon} $ be a quantum channel satisfying that $ LR^{\epsilon} (\mathcal{N}_{AB}) = LR (\mathcal{N}_{AB}^{\varepsilon}) $. We have that
	\begin{align*}
		LR^{\epsilon} (\Theta_{AB\to A'B'}[\mathcal{N}_{AB}]) &\le LR(\Theta_{AB\to A'B'}[\mathcal{N}_{AB}^{\varepsilon}])\\
		&\le LR (\mathcal{N}_{AB}^{\varepsilon}) +\log (1+\delta)\\
		& = LR^{\epsilon} (\mathcal{N}_{AB}) + \log (1+\delta),
	\end{align*}
	concluding the proof.
\end{proof}
\begin{figure}
	\centering
	\tikzset{every picture/.style={line width=0.75pt}} 
	\begin{tikzpicture}[x=0.75pt,y=0.75pt,yscale=-0.7,xscale=0.7]
		
		\draw   (205,91) -- (264,91) -- (264,139) -- (205,139) -- cycle ;
		\draw   (204,157) -- (264,157) -- (264,205) -- (204,205) -- cycle ;
		\draw    (165,101) -- (204,101) ;
		\draw    (166,130) -- (204,130) ;
		\draw    (166,167) -- (204,168) ;
		\draw    (166,196) -- (203,197) ;
		\draw    (263,100) -- (298,100) ;
		\draw    (264,132) -- (298,132) ;
		\draw    (263,166) -- (297,166) ;
		\draw    (264,195) -- (297,195) ;
		\draw    (166,59) -- (166,205) ;
		\draw    (166,59) -- (298,59) ;
		\draw    (298,59) -- (298,205) ;
		\draw    (132,205) -- (166,205) ;
		\draw    (132,20) -- (132,205) ;
		\draw    (331,20) -- (331,205) ;
		\draw    (331,205) -- (298,205) ;
		\draw    (331,20) -- (132,20) ;
		\draw    (99,101) -- (131,101) ;
		\draw    (99,130) -- (131,130) ;
		\draw    (100,167) -- (132,167) ;
		\draw    (99,197) -- (131,197) ;
		\draw    (331,100) -- (360,100) ;
		\draw    (330,129) -- (359,129) ;
		\draw    (331,166) -- (360,166) ;
		\draw    (331,195) -- (360,195) ;
		\draw   (505,88) -- (564,88) -- (564,136) -- (505,136) -- cycle ;
		\draw   (504,154) -- (564,154) -- (564,202) -- (504,202) -- cycle ;
		\draw    (470,98) -- (505,98) ;
		\draw    (470,127) -- (505,127) ;
		\draw    (470,164) -- (505,165) ;
		\draw    (470,193) -- (505,194) ;
		\draw    (563,97) -- (598,97) ;
		\draw    (564,129) -- (598,129) ;
		\draw    (563,163) -- (597,163) ;
		\draw    (564,192) -- (597,192) ;
		
		\draw (207,103) node [anchor=north west][inner sep=0.75pt]   [align=left] {$ \mathcal{F}_{A'B'}^{K} $};
		\draw (211,169) node [anchor=north west][inner sep=0.75pt]   [align=left] {$ \mathcal{F}_{CD}^{L} $};
		\draw (172,105) node [anchor=north west][inner sep=0.75pt]   [align=left] {$ B_{0}' $};
		\draw (171,75) node [anchor=north west][inner sep=0.75pt]   [align=left] {$ A_{0}' $};
		\draw (171,174) node [anchor=north west][inner sep=0.75pt]   [align=left] {$ D_{0} $};
		\draw (172,144) node [anchor=north west][inner sep=0.75pt]   [align=left] {$ C_{0} $};
		\draw (268,105) node [anchor=north west][inner sep=0.75pt]   [align=left] {$ B_{1}' $};
		\draw (268,75) node [anchor=north west][inner sep=0.75pt]   [align=left] {$ A_{1}' $};
		\draw (268,173) node [anchor=north west][inner sep=0.75pt]   [align=left] {$ D_{1} $};
		\draw (268,143) node [anchor=north west][inner sep=0.75pt]   [align=left] {$ C_{1} $};
		\draw (160,30) node [anchor=north west][inner sep=0.75pt]   [align=left] {$\Theta_{A'B'CD\to ABCD}$};
		\draw (365,115) node [anchor=north west][inner sep=0.75pt]   [align=left] {$ B_{1} $};
		\draw (366,85) node [anchor=north west][inner sep=0.75pt]   [align=left] {$ A_{1} $};
		\draw (368,181) node [anchor=north west][inner sep=0.75pt]   [align=left] {$ D_{1} $};
		\draw (367,149) node [anchor=north west][inner sep=0.75pt]   [align=left] {$ C_{1} $};
		\draw (72,86) node [anchor=north west][inner sep=0.75pt]   [align=left] {$ A_{0} $};
		\draw (69,154) node [anchor=north west][inner sep=0.75pt]   [align=left] {$ C_{0} $};
		\draw (71,117) node [anchor=north west][inner sep=0.75pt]   [align=left] {$ B_{0} $};
		\draw (70,183) node [anchor=north west][inner sep=0.75pt]   [align=left] {$ D_{0} $};
		\draw (405,115) node [anchor=north west][inner sep=0.75pt]   [align=left] {$ \approx_{\varepsilon} $};
		\draw (513,100) node [anchor=north west][inner sep=0.75pt]   [align=left] {$ \mathcal{N}_{AB} $};
		\draw (511,166) node [anchor=north west][inner sep=0.75pt]   [align=left] {$ \mathcal{F}_{CD}^{L} $};
		\draw (440,115) node [anchor=north west][inner sep=0.75pt]   [align=left] {$ B_{0} $};
		\draw (440,85) node [anchor=north west][inner sep=0.75pt]   [align=left] {$ A_{0} $};
		\draw (440,179) node [anchor=north west][inner sep=0.75pt]   [align=left] {$ D_{0} $};
		\draw (440,151) node [anchor=north west][inner sep=0.75pt]   [align=left] {$ C_{0} $};
		\draw (602,116) node [anchor=north west][inner sep=0.75pt]   [align=left] {$ B_{1} $};
		\draw (601,85) node [anchor=north west][inner sep=0.75pt]   [align=left] {$ A_{1} $};
		\draw (604,179) node [anchor=north west][inner sep=0.75pt]   [align=left] {$ D_{1} $};
		\draw (603,150) node [anchor=north west][inner sep=0.75pt]   [align=left] {$ C_{1} $};
		\end{tikzpicture}
	\caption{One-shot catalytic dynamic entanglement cost of a bipartite quantum channel $ \mathcal{N}_{AB} $ under $ \delta\mhyphen\seppsc $ $ \Theta_{A'B'CD\to ABCD} $.}
\end{figure}
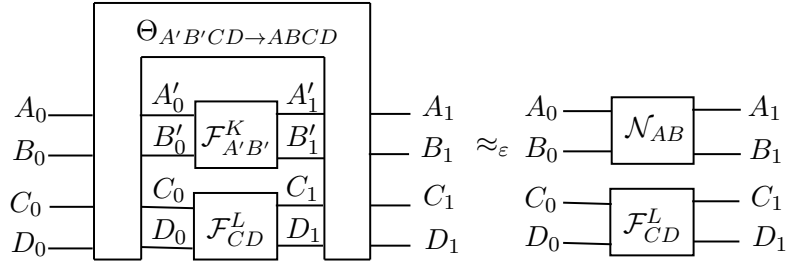
With these tools, we give the formal definition of the one-shot catalytic dynamic entanglement cost of a bipartite channel as follows:
\begin{definition} \label{1-shot_catalytic_cost_def}
	Given $\delta > 0$ and $ \varepsilon\ge 0 $, the one-shot catalytic dynamic entanglement cost of a bipartite quantum channel $ \mathcal{N}_{AB} $ under $ \delta\mhyphen\seppsc $ is defined as
	\begin{align*}
	\widetilde{E}_{C,\delta\mhyphen\seppsc}^{(1), \varepsilon}(\mathcal{N}_{AB}) \coloneqq & \min \left\{ \log K^{2}: \Theta_{A'B'CD\to ABCD}[\mathcal{F}_{A'B'}^{K}\otimes \mathcal{F}_{CD}^{L}] = \mathcal{N}_{AB}'\otimes \mathcal{F}_{CD}^{L},\right.\\
	&\; \Theta_{A'B'CD\to ABCD}\in \delta\mhyphen\seppsc(A'C:B'D\to AC:BD),\\
	&\;\left. \dfrac{1}{2}\dnorm{\mathcal{N}_{AB}' - \mathcal{N}_{AB}}\le \varepsilon, K, L\in \mathbb{N}_{0} \right\}.
	\end{align*}
\end{definition}
In order to bound the one-shot catalytic dynamic entanglement cost of a bipartite channel, the following lemma uses a twisted twirling superchannel that separates the $K$-swap channel from the others.
\begin{lemma}\label{lem.: catalytic channel construction}
	For a bipartite quantum channel $ \mathcal{N}_{AB} $ and $ \varepsilon \ge 0  $, there is a quantum channel $ \mathcal{M}_{ABCD}^{\varepsilon} $ given by
	\begin{equation}\label{key}
	\mathcal{M}_{ABCD}^{\varepsilon} =  p \mathcal{N}_{AB}^{\varepsilon} \otimes \mathcal{F}_{CD}^{L} + (1-p) \mathcal{L}_{ABCD},
	\end{equation}
	where $\mathcal{L}_{ABCD}$ is a quantum channel, $ \frac{1}{2}\pnorm[\diamond]{\mathcal{N}_{AB}^{\varepsilon} - \mathcal{N}_{AB}}\le \abs{A_{0}}\abs{B_{0}}\sqrt{2\varepsilon} $, and $ p\ge 1-2\varepsilon $. It also satisfies that
	\begin{equation}\label{key}
	LR(\mathcal{M}_{ABCD}^{\varepsilon})\le LR^{\varepsilon} (\mathcal{N}_{AB}\otimes \mathcal{F}_{CD}^{L}).
	\end{equation}
\end{lemma}
\begin{proof}
	Let $ \widetilde{\mathcal{M}}_{ABCD}^{\varepsilon} $ be a quantum channel satisfying
	\begin{equation}\label{key}
	LR^{\varepsilon}(\mathcal{N}_{AB}\otimes \mathcal{F}_{CD}^{L}) = LR(\widetilde{\mathcal{M}}_{ABCD}^{\varepsilon})\equiv l,
	\end{equation}
	which implies the existence of a separable channel $ \Sigma_{ABCD}\in \sepc(AC:BD) $ such that
	\begin{equation}\label{key}
	\widetilde{\mathcal{M}}_{ABCD}^{\varepsilon}\le 2^{l} \Sigma_{ABCD}.
	\end{equation}
	Since $ \widetilde{\mathcal{M}}_{ABCD}^{\varepsilon} $ is $ \varepsilon $-close to $ \mathcal{N}_{AB}\otimes \mathcal{F}_{CD}^{L} $ by definition, we expect to have $ \mathcal{M}_{ABCD}^{\varepsilon} $ by properly pinching it. We try the following twisted twirling superchannel, which can be performed via LOCC:
	\begin{align*}
	\Omega_{AB}[\mathcal{E}_{AB}] \coloneqq \iint \mathcal{U}_{A_{1}}\otimes \mathcal{V}_{B_{1}}\circ \mathcal{E}_{AB}\circ \mathcal{V}_{A_{0}}^{\dag}\otimes \mathcal{U}_{B_{0}}^{\dag}.
	\end{align*}
	For a quantum channel $ \mathcal{E}_{AB} $, the twisted twirling superchannel turns its (normalized) Choi matrix into a structured form:
	\begin{align*}
	J_{A_{0}B_{0}\widetilde{A}_{1}\widetilde{B}_{1}}^{\Omega_{AB}[\mathcal{E}_{AB}]}&=\iint \overline{\mathcal{V}}_{A_{0}}\otimes \overline{ \mathcal{U}}_{B_{0}}\otimes \mathcal{U}_{\widetilde{A}_{1}}\otimes \mathcal{V}_{\widetilde{B}_{1}} \left( J_{A_{0}B_{0}\widetilde{A}_{1}\widetilde{B}_{1}}^{\mathcal{E}_{AB}} \right)\\
	&= p_{0}\, \Phi_{A_{0}\widetilde{B}_{1}}^{K}\otimes \Phi_{\widetilde{A}_{1}B_{0}}^{K} + p_{1}\, \Phi_{A_{0}\widetilde{B}_{1}}^{K}\otimes \dfrac{I-\Phi_{\widetilde{A}_{1}B_{0}}^{K}}{K^{2}-1}\\
	&\quad + p_{2}\,\dfrac{I-\Phi_{A_{0}\widetilde{B}_{1}}^{K}}{K^{2}-1}\otimes \Phi_{\widetilde{A}_{1}B_{0}}^{K} + p_{3}\,\dfrac{I-\Phi_{A_{0}\widetilde{B}_{1}}^{K}}{K^{2}-1}\otimes\dfrac{I-\Phi_{\widetilde{A}_{1}B_{0}}^{K}}{K^{2}-1}\\
	&= p_{0} J_{A_{0}B_{0}\widetilde{A}_{1}\widetilde{B}_{1}}^{\mathcal{F}_{AB}^{K}} + (1-p_{0}) J_{A_{0}B_{0}\widetilde{A}_{1}\widetilde{B}_{1}}^{\mathcal{Q}_{AB}}.
	\end{align*}
	Note that $ \tr \left( J_{A_{0}B_{0}\widetilde{A}_{1}\widetilde{B}_{1}}^{\mathcal{F}_{AB}^{K}}J_{A_{0}B_{0}\widetilde{A}_{1}\widetilde{B}_{1}}^{\mathcal{Q}_{AB}} \right) = 0 $.
	Applying the twisted twirling superchannel $ \Omega_{CD} $ on $ \widetilde{\mathcal{M}}_{ABCD}^{\varepsilon} $, we devise $ \mathcal{M}_{ABCD}^{\varepsilon} $ by  as follows:
	\begin{align*}
	\mathcal{M}_{ABCD}^{\varepsilon} &= \Omega_{CD}\left [\widetilde{\mathcal{M}}_{ABCD}^{\varepsilon}\right ] \\
	&= p \mathcal{N}_{AB}^{\varepsilon} \otimes \mathcal{F}_{CD}^{L} + (1-p) \mathcal{L}_{ABCD}.
	\end{align*}
	We show that $ \mathcal{M}_{ABCD}^{\varepsilon} $ satisfies the insisted properties. Firstly, by construction,
	\begin{equation}\label{key}
	\mathcal{M}_{ABCD}^{\varepsilon}=\Omega_{CD}\left [\widetilde{\mathcal{M}}_{ABCD}^{\varepsilon}\right ]\le 2^{l} \Omega_{CD}\left[ \Sigma_{ABCD} \right].
	\end{equation}
	Since $ \Omega_{CD} $ can be done by LOCC, we have $ \Omega_{CD}\left[ \Sigma_{ABCD} \right] \in \sepc(AC:BD)$. Therefore, we have
	\begin{equation}\label{key}
	LR(\mathcal{M}_{ABCD}^{\varepsilon})\le LR^{\varepsilon} (\mathcal{N}_{AB}\otimes \mathcal{F}_{CD}^{L}).
	\end{equation}
	
	From the contractivity of the diamond distance under a superchannel, it follows that
	\begin{align*}
	\varepsilon &\ge \dfrac{1}{2}\pnorm[\diamond]{\widetilde{\mathcal{M}}_{ABCD}^{\varepsilon} - \mathcal{N}_{AB}\otimes \mathcal{F}_{CD}^{L}}\\
	&\ge  \dfrac{1}{2}\pnorm[\diamond]{\Omega_{CD}\left [\widetilde{\mathcal{M}}_{ABCD}^{\varepsilon}\right ] - \mathcal{N}_{AB}\otimes \mathcal{F}_{CD}^{L}},
	\end{align*}
	where we used $ \Omega_{CD}[\mathcal{F}_{CD}^{L}] = \mathcal{F}_{CD}^{L}$. Using Theorem \ref{thm: diamond dist. to fidelity}, we get to
	\begin{align*}
	1-2\varepsilon &\le F\left(J^{\Omega_{CD}\left [\widetilde{\mathcal{M}}_{ABCD}^{\varepsilon}\right ] }, J^{\mathcal{N}_{AB}\otimes \mathcal{F}_{CD}^{L}}  \right)\\
	&= F\left(p J^{\mathcal{N}_{AB}^{\varepsilon} \otimes \mathcal{F}_{CD}^{L} }+ (1-p) J^{\mathcal{L}_{ABCD}}, J^{\mathcal{N}_{AB}\otimes \mathcal{F}_{CD}^{L}}  \right)\\
	&\le p F\left(J^{\mathcal{N}_{AB}^{\varepsilon} \otimes \mathcal{F}_{CD}^{L} }, J^{\mathcal{N}_{AB}\otimes \mathcal{F}_{CD}^{L}}  \right)+ (1-p) F\left(J^{\mathcal{L}_{ABCD}}, J^{\mathcal{N}_{AB}\otimes \mathcal{F}_{CD}^{L}}  \right)\\
	&\le pF\left(J^{\mathcal{N}_{AB}^{\varepsilon} \otimes \mathcal{F}_{CD}^{L} }, J^{\mathcal{N}_{AB}\otimes \mathcal{F}_{CD}^{L}}  \right)\\
	&=pF\left(J^{\mathcal{N}_{AB}^{\varepsilon} }, J^{\mathcal{N}_{AB}}  \right),
	\end{align*}
	where the second inequality follows from the joint concavity of the fidelity, and the third from the orthogonality of the Choi matrices. From the above, we read that $ p\ge 1-2\varepsilon $ and $ F\left(J^{\mathcal{N}_{AB}^{\varepsilon} }, J^{\mathcal{N}_{AB}}  \right) \ge 1-2\varepsilon $ due to $ p\le 1 $ and $ F\left(J^{\mathcal{N}_{AB}^{\varepsilon} }, J^{\mathcal{N}_{AB}}  \right)\le 1 $. Furthermore, $ F\left(J^{\mathcal{N}_{AB}^{\varepsilon} }, J^{\mathcal{N}_{AB}}  \right) \ge 1-2\varepsilon $ together with Theorem \ref{thm: diamond distance - Choi trace distance} and the Fuchs-van der Graaf inequality implies the following:
	\begin{align*}
	\dfrac{1}{2}\pnorm[\diamond]{\mathcal{N}_{AB}^{\varepsilon} - \mathcal{N}_{AB}} &\le \abs{A_{0}}\abs{B_{0}}\, \dfrac{1}{2}\pnorm[1]{J^{\mathcal{N}_{AB}^{\varepsilon}} - J^{\mathcal{N}_{AB}}}\\
	&\le \abs{A_{0}}\abs{B_{0}}\sqrt{1 - F\left( J^{\mathcal{N}_{AB}^{\varepsilon}}, J^{\mathcal{N}_{AB}} \right)}\\
	&\le \abs{A_{0}}\abs{B_{0}} \sqrt{2 \varepsilon}.
	\end{align*}
	This completes the proof.
\end{proof}

We bound the one-shot catalytic dynamic entanglement cost of a bipartite channel as follows:
\begin{theorem}
	Given $ \delta > 0 $, $ \varepsilon \ge 0 $, there exists $ L\in \mathbb{N} $ such that $ L^{2}\ge 1+ \frac{1}{\delta}$, and the one-shot catalytic dynamic entanglement cost for any bipartite quantum channel $ \mathcal{N}_{AB} $ is bounded as
	\begin{align*}
	LR^{\varepsilon}(\mathcal{N}_{AB}&\otimes \mathcal{F}_{CD}^{L})  - \log L^{2} -\log (1+\delta)  \\
	&\le \widetilde{E}_{C,\delta\mhyphen\seppsc}^{(1), \varepsilon}(\mathcal{N}_{AB}) \\
	& \qquad\qquad \le LR^{\varepsilon'}(\mathcal{N}_{AB}\otimes \mathcal{F}_{CD}^{L}) - \log L^{2}-\log(1-2\varepsilon') + 2,
	\end{align*}
	where $ \varepsilon' = \varepsilon^{2}/\left( 2\abs{A_{0}}^{2}\abs{B_{0}}^{2} \right) $.
\end{theorem}
\begin{proof}
We break down the argument into separate proofs of the two bounds.
\begin{enumerate}[(i)]
    \item For the upper bound, let $ \mathcal{M}_{ABCD}^{\varepsilon} $ be a quantum channel as in Lemma \ref{lem.: catalytic channel construction} satisfying that
	\begin{align*}
	LR (\mathcal{M}_{ABCD}^{\varepsilon'}) &\le LR^{\varepsilon'} (\mathcal{N}_{AB}\otimes \mathcal{F}_{CD}^{L}),\\
	\mathcal{M}_{ABCD}^{\varepsilon'}&= p \mathcal{N}_{AB}^{\varepsilon} \otimes \mathcal{F}_{CD}^{L} + (1-p) \mathcal{L}_{ABCD},\\
	\dfrac{1}{2}\pnorm[\diamond]{\mathcal{N}_{AB}^{\varepsilon} - \mathcal{N}_{AB}}&\le \varepsilon,\\
	p&\ge 1-2\varepsilon',
	\end{align*}
	where $ \varepsilon' = \varepsilon^{2}/\left( 2\abs{A_{0}}^{2}\abs{B_{0}}^{2} \right) $.
	From the first and the second equation above, it follows that
	\begin{equation*}
	\mathcal{M}_{ABCD}^{\varepsilon'}= p \mathcal{N}_{AB}^{\varepsilon} \otimes \mathcal{F}_{CD}^{L} + (1-p) \mathcal{L}_{ABCD}\le 2^{LR^{\varepsilon'} (\mathcal{N}_{AB}\otimes \mathcal{F}_{CD}^{L})}\Sigma_{ABCD},
	\end{equation*}
	where $ \Sigma_{ABCD}\in\sepc(AC:BD) $. Since $ \mathcal{L}_{ABCD}\ge 0 $, we have that
	\begin{align*}
	\mathcal{N}_{AB}^{\varepsilon} \otimes \mathcal{F}_{CD}^{L}&\le 2^{LR^{\varepsilon'} (\mathcal{N}_{AB}\otimes \mathcal{F}_{CD}^{L}) -\log p}\Sigma_{ABCD}\\
	&\le 2^{LR^{\varepsilon'} (\mathcal{N}_{AB}\otimes \mathcal{F}_{CD}^{L}) -\log (1-2\varepsilon')}\Sigma_{ABCD},
	\end{align*}
	which leads to the existence of a quantum channel $ \mathcal{R}_{ABCD} $ such that
	\begin{align}\label{key}
	&\mathcal{N}_{AB}^{\varepsilon} \otimes \mathcal{F}_{CD}^{L} + \left\{ 2^{LR^{\varepsilon'} (\mathcal{N}_{AB}\otimes \mathcal{F}_{CD}^{L}) -\log (1-2\varepsilon')}-1 \right\} \mathcal{R}_{ABCD}\\
	=\;&\mathcal{N}_{AB}^{\varepsilon} \otimes \mathcal{F}_{CD}^{L} + ( r - 1) \mathcal{R}_{ABCD} \label{eq: robustness of N times F}\\
	\propto\;& \sepc(AC:BD),
	\end{align}
	where we denote $ r = 2^{LR^{\varepsilon'} (\mathcal{N}_{AB}\otimes \mathcal{F}_{CD}^{L}) -\log (1-2\varepsilon')}$.
	With the insight gained above, we construct a superchannel $ \Theta_{A'B'CD\to ABCD} \in\delta\mhyphen\seppsc (A'C:B'D\to AC:BD)$ as follows:
	\begin{align*}
	\Theta_{A'B'CD\to ABCD}[\mathcal{E}_{A'B'CD}]&\coloneqq \tr \left( J^{\mathcal{F}_{A'B'}^{K}\otimes \mathcal{F}_{CD}^{L}} J^{\mathcal{E}_{A'B'CD}} \right) \mathcal{N}_{AB}^{\varepsilon}\otimes \mathcal{F}_{CD}^{L}\\
	&\quad + \tr \left\{ \left( I - J^{\mathcal{F}_{A'B'}^{K}\otimes \mathcal{F}_{CD}^{L}}\right) J^{\mathcal{E}_{A'B'CD}}  \right\} \mathcal{R}_{ABCD},
	\end{align*}
	where we use the (normalized) Choi matrix and the identity matrix omitting some of the indices for brevity. It is obvious that $ \Theta_{A'B'CD\to ABCD}[\mathcal{F}_{A'B'}^{K}\otimes \mathcal{F}_{CD}^{L}]=\mathcal{N}_{AB}^{\varepsilon}\otimes \mathcal{F}_{CD}^{L} $ from the construction. We now show that $ \Theta_{A'B'CD\to ABCD}\in \delta\mhyphen\seppsc (A'C:B'D\to AC:BD) $.
	For $ \mathcal{E}_{A'B'CD}\in\sepc(A'C:B'D) $, we have that
	\begin{align*}
	\Theta_{A'B'CD\to ABCD}[\mathcal{E}_{A'B'CD}]&= \tr \left( J^{\mathcal{F}_{A'B'}^{K}\otimes \mathcal{F}_{CD}^{L}} J^{\mathcal{E}_{A'B'CD}} \right) \mathcal{N}_{AB}^{\varepsilon}\otimes \mathcal{F}_{CD}^{L}\\
	&\quad + \tr \left\{ \left( I - J^{\mathcal{F}_{A'B'}^{K}\otimes \mathcal{F}_{CD}^{L}}\right) J^{\mathcal{E}_{A'B'CD}}  \right\} \mathcal{R}_{ABCD}\\
	&= q\,\dfrac{\mathcal{N}_{AB}^{\varepsilon} \otimes \mathcal{F}_{CD}^{L} + ( r - 1) \mathcal{R}_{ABCD} }{r} + (1-q) \mathcal{R}_{ABCD},
	\end{align*}
	where $ q = r\tr \left( J^{\mathcal{F}_{A'B'}^{K}\otimes \mathcal{F}_{CD}^{L}} J^{\mathcal{E}_{A'B'CD}} \right) \le \frac{r}{K^{2}L^{2}}$ due to Proposition~\ref{thm: max fidelity btw. ent. sep. states}. Setting $ K =  \left \lceil \frac{\sqrt{r}}{L} \right \rceil $, it is assured that $ q\le 1 $. Therefore, from the convexity of the robustness, it follows that
	\begin{align*}
	R\left(  \Theta_{A'B'CD\to ABCD}[\mathcal{E}_{A'B'CD}]\right)&\le q\,R\left( \dfrac{\mathcal{N}_{AB}^{\varepsilon} \otimes \mathcal{F}_{CD}^{L} + ( r - 1) \mathcal{R}_{ABCD} }{r} \right) + (1-q) R\left(\mathcal{R}_{ABCD}  \right)\\
	&\le R\left(\mathcal{R}_{ABCD}  \right).
	\end{align*}
	Furthermore, we have that
	\begin{equation}\label{key}
	R\left(\mathcal{R}_{ABCD}  \right)\le \dfrac{1}{R\left( \mathcal{N}_{AB}^{\varepsilon}\otimes \mathcal{F}_{CD}^{L} \right)}\le \dfrac{1}{R\left( \mathcal{F}_{CD}^{L} \right)} = \dfrac{1}{L^{2}-1},
	\end{equation}
	where the first inequality follows from equation \eqref{eq: robustness of N times F}, and the second inequality follows from the monotonicity of the robustness under \seppsc,\footnote{One can feed any product state $\rho_{A}\otimes \rho_{B}$ into the quantum channel and subsequently trace away some subsystems at the output.} that is, $ R\left( \mathcal{N}_{AB}^{\varepsilon}\otimes \mathcal{F}_{CD}^{L} \right)\ge R\left( \mathcal{F}_{CD}^{L} \right)$.
	Thus, if $ R\left(\mathcal{R}_{ABCD}  \right)\le \frac{1}{L^{2}-1}\le \delta $, or $ L^{2}\ge 1+\frac{1}{\delta} $, then for $ \mathcal{E}_{A'B'CD}\in\sepc(A'C:B'D) $, we have that
	\begin{equation}
	R\left(  \Theta_{A'B'CD\to ABCD}[\mathcal{E}_{A'B'CD}]\right)\le R\left(\mathcal{R}_{ABCD}  \right)\le \delta.
	\end{equation}
	So we have that $ \Theta_{A'B'CD\to ABCD} \in \delta\mhyphen\seppsc (A'C:B'D\to AC:BD)$ by setting $K$ as
	\begin{equation}\label{key}
	K =  \left \lceil \frac{\sqrt{r}}{L} \right \rceil =\left \lceil\dfrac{\sqrt{ 2^{LR^{\varepsilon'} (\mathcal{N}_{AB}\otimes \mathcal{F}_{CD}^{L}) -\log (1-2\varepsilon')}}}{L}\right \rceil,
	\end{equation}
	where $L\in \mathbb{N}$ is chosen to satisfy that $ L^{2}\ge 1+\frac{1}{\delta} $.
	Finally, we conclude that
	\begin{align*}
	\widetilde{E}_{C, \delta\mhyphen\seppsc}^{(1),\varepsilon}(\mathcal{N}_{AB})&\le \log K^{2}\\
	&\le LR^{\varepsilon'}\left( \mathcal{N}_{AB}\otimes \mathcal{F}_{CD}^{L} \right) - \log L^{2} - \log (1-2\varepsilon') + 2,
	\end{align*}
	where $ \varepsilon' = \varepsilon^{2}/\left( 2\abs{A_{0}}^{2}\abs{B_{0}}^{2} \right) $.
	
	\item For the lower bound, let $ \widetilde{E}_{C,\delta\mhyphen\seppsc}^{(1), \varepsilon}(\mathcal{N}_{AB}) = \log K^{2} $ for which a catalyst $ \mathcal{F}_{CD}^{L_{0}} $ is used as follows:
	\begin{equation}\label{key}
	\Theta_{A'B'CD\to ABCD}[\mathcal{F}_{A'B'}^{K}\otimes \mathcal{F}_{CD}^{L_{0}}] = \mathcal{N}_{AB}'\otimes \mathcal{F}_{CD}^{L_{0}}, \quad\mathcal{N}_{AB}'\approx_{\varepsilon} \mathcal{N}_{AB},
	\end{equation}
	where $ \Theta_{A'B'CD\to ABCD}\in \delta\mhyphen\seppsc(A'C:B'D\to AC:BD) $.
	It follows that
	\begin{align*}
	LR^{\varepsilon}\left( \mathcal{N}_{AB}\otimes \mathcal{F}_{CD}^{L_{0}} \right) &\le LR\left( \mathcal{N}_{AB}'\otimes \mathcal{F}_{CD}^{L_{0}} \right)\\
	&= LR\left( \Theta_{A'B'CD\to ABCD}[\mathcal{F}_{A'B'}^{K}\otimes \mathcal{F}_{CD}^{L_{0}}] \right)\\
	&\le LR\left( \mathcal{F}_{A'B'}^{K}\otimes \mathcal{F}_{CD}^{L_{0}} \right) + \log (1+\delta)\\
	& = \log K^{2} + \log L_{0}^{2} + \log (1+\delta).
	\end{align*}
    Moreover, we can choose a universal lower bound $ \widetilde{L}_{0} $ that satisfies the following:
    \begin{equation*}
	    LR^{\varepsilon}(\mathcal{N}_{AB}\otimes \mathcal{F}_{CD}^{\widetilde{L}_{0}})  - \log \widetilde{L}_{0}^{2} -\log (1+\delta) \le \widetilde{E}_{C,\delta\mhyphen\seppsc}^{(1), \varepsilon}(\mathcal{N}_{AB})\quad \forall \mathcal{N}_{AB}\in \cptp(AB).
	\end{equation*}
\end{enumerate}
Finally, regarding (i) and (ii), we can choose $L = \max\left\{\widetilde{L}_{0}, \left\lceil \sqrt{1+\frac{1}{\delta}} \right\rceil\right \}$ which provides both the upper and the lower bound on $ \widetilde{E}_{C,\delta\mhyphen\seppsc}^{(1), \varepsilon}(\mathcal{N}_{AB}) $ for any bipartite quantum channel $ \mathcal{N}_{AB}) $. This completes the proof.
\end{proof}

\section{Conclusion}
We found that entanglement of quantum channels can be naturally understood adopting the superchannel framework. Taking the separable channels as our free resource, we defined the separability-preserving superchannels as the resource non-generating superchannels. The $K$-swap channel $ \mathcal{F}_{AB}^{K} $ is chosen as the dynamic entanglement resource, mimicking the role of the $K$-maximally entangled state in the resource theory of static entanglement. In fact, these two objects are totally interchangeable, because a $K$-swap channel can be transformed into a pair of $K$-maximally entangled states under LOCC, and vice versa two $K$-maximally entangled states can be used to implement a $K$-swap channel with LOCC --- more precisely, by performing two times a teleportation protocol.
Our results provide an operational meaning to the standard and the generalized log-robustness of channels as well as the hypothesis-testing relative entropy of dynamic entanglement that we constructed from the hypothesis-testing relative entropy of channels with minimization over the set of separable channels: The one-shot dynamic entanglement cost can be bounded by the standard log-robustness of channels with respect to the separable channels. The one-shot distillable dynamic entanglement is bounded by the hypothesis-testing relative entropy of dynamic entanglement. When it comes to the catalytic scenario where additional dynamic entanglement resource is supplied and returned back after the free superchannel, we find that the one-shot catalytic dynamic  entanglement cost is bounded by the generalized log-robustness of channels with respect to the set of separable channels. Finally, in the appendices, we investigate the asymptotic scenario, using the liberal dynamic entanglement cost of a bipartite quantum channel, which features the liberal smoothing instead of the diamond norm smoothing. It is shown that the quantity is equal to the liberal regularized relative entropy of channels minimized over the separable channels.

\section*{Acknowledgement}
This research was supported by the National Research Foundation of Korea grant funded by the Ministry of Science and ICT (MSIT) (Grant no.\ NRF-2019R1A2C1006337) and (Grant no.\ NRF-2020M3E4A1079678). S.L.\ acknowledges support from the MSIT, Korea, under the Information Technology Research Center support program (IITP-2020-2018-0-01402) supervised by the Institute for Information \& Communications Technology Planning \& Evaluation. L.L.\ acknowledges support from the Alexander von Humboldt Foundation.

\section*{Note added}
During the completion of this manuscript, we became aware of two independent works on dynamic resource theories: B. Regula and R. Takagi \cite{regula2020OneshotManipulationDynamical} formulated one-shot manipulation of dynamic resources in a general setting, while X. Yuan, et al. \cite{xyuan2020OneshotDynamicalResource} also investigated one-shot distillation and dilution of dynamic resources in a general setting.

\appendix
\section{Appendix}
\subsection{Liberal Dynamic Entanglement Cost of a Bipartite Channel}
There are several alternative ways of smoothing in channel resource theories utilized in the study of the asymptotic equipartition properties \cite{gour2019HowQuantifyDynamical}. For a quantum channel $ \mathcal{N}_{A} $ and a quantum state $ \varphi_{A_{0}R_{0}} $, we denote the $ \varepsilon $-diamond ball and the $ \varepsilon $-liberal ball as
\begin{align*}
	B_{\varepsilon}(\mathcal{N}_{A}) &\coloneqq \left \{ \mathcal{N}_{A}'\in\cptp(A): \dfrac{1}{2}\dnorm{\mathcal{N}_{A}' - \mathcal{N}_{A}}\le \varepsilon \right \},\\
	B_{\varepsilon}^{\varphi}(\mathcal{N}_{A}) &\coloneqq \left \{ \mathcal{N}_{A}'\in\cptp(A): \dfrac{1}{2}\pnorm[1]{\mathcal{N}_{A}'(\varphi_{A_{0}R_{0}}) - \mathcal{N}_{A}(\varphi_{A_{0}R_{0}})}\le \varepsilon \right \}.
\end{align*}
Observe that $ B_{\varepsilon}(\mathcal{N}_{A})\subset \cap_{\varphi_{A_{0}R_{0}}} B_{\varepsilon}^{\varphi}(\mathcal{N}_{A}) $.
For a set of free resource $ \mathcal{F} $, the relevant liberal quantities are defined as follows:
\begin{align*}
	LR_{\mathcal{F}}^{\varepsilon,\varphi}(\mathcal{N}_{A}) &\coloneqq \min_{\mathcal{N}_{A}'\in B_{\varepsilon}^{\varphi}(\mathcal{N}_{A})}LR_{\mathcal{F}}(\mathcal{N}_{A}'),\\
	LR_{\mathcal{F}}^{\varepsilon}(\mathcal{N}_{A}) &\coloneqq \max_{\varphi_{A_{0}R_{0}}} \min_{\mathcal{N}_{A}'\in B_{\varepsilon}^{\varphi}(\mathcal{N}_{A})}LR_{\mathcal{F}}(\mathcal{N}_{A}'),\\
	LR_{\mathcal{F}}^{\varepsilon, n}(\mathcal{N}_{A}) &\coloneqq \dfrac{1}{n}\max_{\varphi_{A_{0}R_{0}}}LR_{\mathcal{F}}^{\varepsilon, \varphi^{\otimes n}}(\mathcal{N}_{A}^{\otimes n}),\\
	LR_{\mathcal{F}}^{(\infty)}(\mathcal{N}_{A}) &\coloneqq \lim_{\varepsilon \to 0^{+}} \liminf_{n\to \infty} LR_{\mathcal{F}}^{\varepsilon, n}(\mathcal{N}_{A}).
\end{align*}
The liberal regularized relative entropy of a channel $ \mathcal{N}_{A} $ with respect to a free resource $ \mathcal{F} $ is defined as
\begin{equation}\label{key}
	D_{\mathcal{F}}^{(\infty)}(\mathcal{N}_{A})\coloneqq \lim_{n\to \infty} \dfrac{1}{n} \max_{\varphi_{A_{0}R_{0}}} \min_{\mathcal{M}\in\mathcal{F}} D\left( \mathcal{N}_{A}^{\otimes n}\left (\varphi_{A_{0}R_{0}}^{\otimes n}\right )\Vert \mathcal{M}_{A}^{\otimes n}\left (\varphi_{A_{0}R_{0}}^{\otimes n}\right )\right).
\end{equation}
It is shown in \cite{gour2019HowQuantifyDynamical} that the asymptotic equipartition property holds as
\begin{equation}\label{key}
	LR_{\mathcal{F}}^{(\infty)} (\mathcal{N}_{A}) = D_{\mathcal{F}}^{(\infty)}(\mathcal{N}_{A}).
\end{equation}

\begin{definition}
	Given $\varepsilon\ge 0$, the $ \varepsilon $-liberal one-shot dynamic entanglement cost of a bipartite quantum channel $ \mathcal{N}_{AB} $ under SEPPSC is defined as
	\begin{equation}\label{key}
		E_{C_{l},\seppsc}^{(1), \varepsilon}(\mathcal{N}_{AB}) \coloneqq \max_{\varphi_{A_{0}A_{0}'B_{0}B_{0}'}} E_{C_{l},\seppsc}^{(1), \varepsilon, \varphi}(\mathcal{N}_{AB}),
	\end{equation}
	where
	\begin{equation}\label{key}
		E_{C_{l},\seppsc}^{(1), \varepsilon, \varphi}(\mathcal{N}_{AB}) \coloneqq \min_{\mathcal{N}_{AB}'\in B_{\varepsilon}^{\varphi}(\mathcal{N}_{AB})} E_{C,\seppsc}^{(1),0}(\mathcal{N}_{AB}').
	\end{equation}
	The liberal (asymptotic) dynamic entanglement cost of a bipartite quantum channel $ \mathcal{N}_{AB} $ under SEPPSC is defined as
	\begin{equation}\label{key}
		E_{C_{l},\seppsc}(\mathcal{N}_{AB}) \coloneqq \lim_{\varepsilon\to 0^{+}}\liminf_{n\to \infty} \dfrac{1}{n} \max_{\varphi_{A_{0}A_{0}'B_{0}B_{0}'}} E_{C_{l},\seppsc}^{(1), \varepsilon, \varphi^{\otimes n}}(\mathcal{N}_{AB}^{\otimes n}).
	\end{equation}
\end{definition}
While an operational meaning of the above quantity is missing yet, it is given by the liberal regularized relative entropy:
\begin{theorem}
	The liberal dynamic entanglement cost of a bipartite quantum channel $ \mathcal{N}_{AB} $ is given by
	\begin{equation}\label{key}
		E_{C_{l},\seppsc}(\mathcal{N}_{AB}) = D_{\sepc}^{(\infty)}(\mathcal{N}_{AB}).
	\end{equation}
\end{theorem}
\begin{proof}
	From Theorem \ref{thm: one-shot dyn. ent. cost}, for any $ \varepsilon\ge 0 $ and $ \varphi $ it holds that
	\begin{equation}\label{key}
		LR_{\sepc}^{\varepsilon,\varphi}(\mathcal{N}_{AB}) \le E_{C_{l},\seppsc}^{(1),\varepsilon,\varphi}(\mathcal{N}_{AB})\le LR_{\sepc}^{\varepsilon,\varphi}(\mathcal{N}_{AB}) +2.
	\end{equation}
	The asymptotic equipartition property leads to the conclusion.
\end{proof}

\subsection{A Few Technical Results}
\begin{proposition}\label{thm: max fidelity btw. ent. sep. states}
   Let $ \ket{\Phi^{K}}_{A_{0}B_{0}}=\frac{1}{\sqrt{K}}\sum_{i=0}^{K-1}\ket{ii}_{A_{0}B_{0}} $ be a $ K $-maximally entangled state where $ \abs{A_{0}}\equiv\dim A_{0}\ge K $ and $ \abs{B_{0}}\equiv\dim B_{0}\ge K $. We have that
	\begin{equation}
	\max_{\sigma_{A_{0}B_{0}}\in \sepd(A_{0}:B_{0})} \tr \Phi_{A_{0}B_{0}}^{K} \sigma_{A_{0}B_{0}} = \dfrac{1}{K}.
	\end{equation}
\end{proposition}
\begin{proof}
	A separable state $ \sigma_{A_{0}B_{0}} $ can be written as a convex sum of pure product states $ \sigma_{A_{0}B_{0}} = \sum_{\alpha}p_{\alpha}\phi_{\alpha}\otimes \psi_{\alpha} $:
	\begin{align*}
	\tr \Phi_{A_{0}B_{0}}^{K} \sigma_{A_{0}B_{0}} &= \dfrac{1}{K}\sum_{i,j=0}^{K-1}\sum_{\alpha} p_{\alpha} \bracket{i}{\phi_{\alpha}}\bracket{i}{\psi_{\alpha}}\bracket{\phi_{\alpha}}{j}\bracket{\psi_{\alpha}}{j}\\
	&=\dfrac{1}{K}\sum_{\alpha}p_{\alpha} \left\lvert \sum_{i=0}^{K-1} \bracket{i}{\phi_{\alpha}}\bracket{i}{\psi_{\alpha}} \right\rvert^{2}\\
	&\le \dfrac{1}{K}\sum_{\alpha} p_{\alpha} \left\{ \sum_{i=0}^{K-1} \abs{\bracket{i}{\phi_{\alpha}}}^{2} \right\} \left\{ \sum_{i=0}^{K-1} \abs{\bracket{i}{\psi_{\alpha}}}^{2} \right\}\\
	&\le \dfrac{1}{K},
	\end{align*}
	where the Cauchy-Schwarz inequality is used for the first inequality.
\end{proof}

\begin{theorem}[\cite{diamond_trace_dist_rel}]\label{thm: diamond distance - Choi trace distance}
	Let $ \mathcal{N}_{A} $ and $ \mathcal{M}_{A} $ be quantum channels, and $ J^{\mathcal{N}_{A}} $ and $ J^{\mathcal{M}_{A}} $ be their (normalized) Choi matrices, respectively. It holds that
	\begin{equation}\label{key}
	\dfrac{1}{\abs{A}}\pnorm[\diamond]{\mathcal{N}_{A} - \mathcal{M}_{A}}\le \pnorm[1]{J^{\mathcal{N}_{A}} - J^{\mathcal{M}_{A}}}\le \pnorm[\diamond]{\mathcal{N}_{A} - \mathcal{M}_{A}}.
	\end{equation}
\end{theorem}
\begin{proof}
	The second inequality follows from the definition of the diamond norm.	For the first inequality, let $ \Psi _{A_{0}R_{0}}$ be the optimal pure state for the diamond distance as
	\begin{equation}\label{key}
		\dnorm{\mathcal{N}_{A} - \mathcal{M}_{A}} = \pnorm[1]{\mathcal{N}_{A}\otimes \mathsf{id}_{R_{0}}(\Psi_{A_{0}R_{0}}) - \mathcal{M}_{A}\otimes \mathsf{id}_{R_{0}}(\Psi_{A_{0}R_{0}})},
	\end{equation}
	where $ \abs{A_{0}}  = \abs{R_{0}}$. One can denote $ \ket{\Psi}_{A_{0}R_{0}} = I_{A_{0}}\otimes X_{R_{0}}\ket{\phi^{+}}_{A_{0}R_{0}} $, where $ \ket{\phi^{+}}_{A_{0}R_{0}} = \sum_{i=0}^{\abs{A_{0}}-1}\ket{ii}_{A_{0}R_{0}} $ and the operator $ X_{R_{0}} $ satisfies $ \tr_{A_{0}R_{0}}\Psi_{A_{0}R_{0}} = 1 = \tr_{R_{0}}X_{R_{0}}^{\dag}X_{R_{0}}=\pnorm[2]{X_{R_{0}}}^{2} $. With $ \Phi_{A_{0}R_{0}}^{+} = \frac{\phi_{A_{0}R_{0}}^{+}}{\abs{A_{0}}} $, we have
	\begin{align*}
		\dnorm{\mathcal{N}_{A} - \mathcal{M}_{A}} & = \pnorm[1]{\mathcal{N}_{A}\otimes \mathsf{id}_{R_{0}}(\Psi_{A_{0}R_{0}}) - \mathcal{M}_{A}\otimes \mathsf{id}_{R_{0}}(\Psi_{A_{0}R_{0}})}\\
		&=\pnorm[1]{\left( \mathcal{N}_{A} - \mathcal{M}_{A} \right)\otimes \mathsf{id}_{R_{0}} \left( I_{A_{0}}\otimes X_{R_{0}}\cdot \abs{A_{0}}\Phi_{A_{0}R_{0}}^{+} \cdot I_{A_{0}}\otimes X_{R_{0}}^{\dag} \right)}\\
		&=\abs{A_{0}} \pnorm[1]{ I_{A_{0}}\otimes X_{R_{0}}\cdot \left( J^{\mathcal{N}_{A}} - J^{\mathcal{M}_{A}} \right) \cdot I_{A_{0}}\otimes X_{R_{0}}^{\dag} }\\
		&\le \abs{A_{0}} \pnorm[\infty]{X_{R_{0}}}\pnorm[\infty]{X_{R_{0}}^{\dag}}\pnorm[1]{J^{\mathcal{N}_{A}} - J^{\mathcal{M}_{A}}}\\
		&\le \abs{A_{0}} \pnorm[2]{X_{R_{0}}}\pnorm[2]{X_{R_{0}}^{\dag}}\pnorm[1]{J^{\mathcal{N}_{A}} - J^{\mathcal{M}_{A}}}\\
		&\le \abs{A_{0}}\pnorm[1]{J^{\mathcal{N}_{A}} - J^{\mathcal{M}_{A}}},
	\end{align*}
where we used the H\"older inequality for the first inequality.
\end{proof}

\begin{theorem}\label{thm: diamond dist. to fidelity}
	Let $ \mathcal{N}_{A} $ and $ \mathcal{M}_{A} $ be quantum channels. Given $ \varepsilon \ge 0 $, we have
	\begin{equation}\label{key}
	\dfrac{1}{2}\pnorm[\diamond]{\mathcal{N}_{A}-\mathcal{M}_{A}} \le \varepsilon \implies \min_{\Psi_{A_{0}R_{0}}} F(\mathcal{N}_{A}(\Psi_{A_{0}R_{0}}), \mathcal{M}_{A}(\Psi_{A_{0}R_{0}}))  \ge (1-\varepsilon)^{2}\ge 1-2\varepsilon.
	\end{equation}
	Conversely, it follows that
	\begin{equation}\label{key}
	\min_{\Psi_{A_{0}R_{0}}}F(\mathcal{N}_{A}(\Psi_{A_{0}R_{0}}), \mathcal{M}_{A}(\Psi_{A_{0}R_{0}}))  \ge 1-\varepsilon \implies \dfrac{1}{2}\pnorm[\diamond]{\mathcal{N}_{A}-\mathcal{M}_{A}} \le \sqrt{\varepsilon}.
	\end{equation}
\end{theorem}
\begin{proof}
	Both follow from the Fuchs-van de Graaf inequality, while $ (1-\varepsilon)^{2}=1-2\varepsilon +\varepsilon^{2}\ge 1-2\varepsilon$ for $ \varepsilon \ge 0 $.
\end{proof}
\begin{proposition}
	Let $ \mathcal{N}_{A} $ and $ \mathcal{M}_{A} $ be quantum channels, and $ J^{\mathcal{N}_{A}} $ and $ J^{\mathcal{M}_{A}} $ be their (normalized) Choi matrices, respectively. If $ F\left( J^{\mathcal{N}_{A}},J^{\mathcal{M}_{A}} \right) \ge 1-\varepsilon $, then $ \frac{1}{2}\pnorm[\diamond]{\mathcal{N}_{A}-\mathcal{M}_{A}}\le n\sqrt{\varepsilon} $.
\end{proposition}
\begin{proof}
	From Theorem \ref{thm: diamond distance - Choi trace distance} and the Fuchs-van der Graaf inequality, it follows that
	\begin{align*}
	\dfrac{1}{2}\pnorm[\diamond]{\mathcal{N}_{A}-\mathcal{M}_{A}}&\le n\, \dfrac{1}{2}\pnorm[1]{J^{\mathcal{N}_{A}} - J^{\mathcal{M}_{A}}}\\
	&\le n \sqrt{1 - F\left( J^{\mathcal{N}_{A}},  J^{\mathcal{M}_{A}} \right)}\\
	&\le n\sqrt{\varepsilon}.
	\end{align*}
\end{proof}

\printbibliography
\end{document}